\documentclass[12pt]{article}
\usepackage{amsmath}
\usepackage{amssymb}
\usepackage{amsthm}
\usepackage{tikz}
\usepackage{enumerate}
\usetikzlibrary{decorations.markings,decorations.pathreplacing,calc}
\usetikzlibrary{arrows}
\usepackage{algorithm}
\usepackage{algpseudocode}
\usepackage{thm-restate}
\usepackage{verbatim}
\usepackage[framemethod=tikz]{mdframed}
\usepackage{complexity}
\usepackage{xparse}
\usepackage[margin=1in]{geometry}
\usepackage{subfig}
\usepackage[utf8]{inputenc}

\usepackage[colorlinks = true]{hyperref}

\usepackage{xcolor}
\definecolor{darkred}  {rgb}{0.5,0,0}
\definecolor{darkblue} {rgb}{0,0,0.5}
\definecolor{darkgreen}{rgb}{0,0.5,0}

\hypersetup{
	urlcolor   = blue,         
	linkcolor  = darkblue,     
	citecolor  = darkgreen,    
	filecolor  = darkred       
}

\newcommand{\be}{\begin{equation}}
	\newcommand{\ee}{\end{equation}}
\newcommand{\ba}{\begin{array}}
	\newcommand{\ea}{\end{array}}
\newcommand{\bea}{\begin{eqnarray}}
	\newcommand{\eea}{\end{eqnarray}}

\renewcommand{\>}{\rangle}
\newcommand{\<}{\langle}

\newcommand{\calA}{{\cal A }}
\newcommand{\calH}{{\cal H }}

\newcommand{\calE}{{\cal E }}

\newcommand{\calD}{{\cal D }}
\newcommand{\calS}{{\cal S }}

\newcommand{\calQ}{{\cal Q }}
\newcommand{\calR}{{\cal R }}

\newcommand{\calZ}{{\cal Z}}
\newcommand{\calW}{{\cal W}}

\newcommand{\CC}{\mathbb{C}}

\newcommand{\RR}{\mathbb{R}}
\newcommand{\la}{\langle}
\newcommand{\ra}{\rangle}

\newcommand{\eps}{\epsilon}

\DeclareMathOperator*{\Ex}{\mathbb{E}}

\newcommand{\zee}{{\cal Z}}

\NewDocumentCommand{\trace}{m o}{%
	\ \mathrm{Tr}#1 %
	\IfValueT{#2}{\ensuremath{_{#2}}}%
}

\DeclareDocumentCommand\outer{ s m g }
{ 
	\IfBooleanTF{#1}
	{ 
		\IfNoValueTF{#3}
		{\vphantom{#2}\left\lvert\smash{#2}\middle\rangle\!\middle\langle\smash{#2}\right\rvert}
		{\vphantom{#2#3}\left\lvert\smash{#2}\middle\rangle\!\middle\langle\smash{#3}\right\rvert}
	}
	{ 
		\IfNoValueTF{#3}
		{\left\lvert{#2}\middle\rangle\!\middle\langle{#2}\right\rvert}
		{\left\lvert{#2}\middle\rangle\!\middle\langle{#3}\right\rvert}
	}
}

\newcommand{\ket}[1]{\left| #1\right\rangle}        

\newcommand{\bra}[1]{\left\langle #1\right|}        
\newcommand{\ketbra}[1]{| #1 \rangle \! \mspace{1mu}\langle #1| }

\newtheorem{dfn}{Definition}
\newtheorem{prop}{Proposition}

\newtheorem{claim}{Claim}
\newtheorem{lemma}{Lemma}
\newtheorem{corol}{Corollary}
\newtheorem{fact}{Fact}

\newtheorem{theorem}{Theorem}
\newtheorem*{theorem*}{Theorem}
\newtheorem{problem}{Problem}
\newtheorem{rmk}{Remark}
\newtheorem*{problem*}{Problem}

\newcommand{\footremember}[2]{%
	\footnote{#2}
	\newcounter{#1}
	\setcounter{#1}{\value{footnote}}%
}
\newcommand{\footrecall}[1]{%
	\footnotemark[\value{#1}]%
}
\title{On the complexity of quantum partition functions}

\author{Sergey Bravyi\footremember{ibm}{IBM Quantum, IBM T.J. Watson Research Center}
	\and
	Anirban Chowdhury \footremember{iqc}{Institute for Quantum Computing, University of Waterloo, Canada}\footremember{co}{Department of Combinatorics and Optimization, University of Waterloo, Canada}
	\and David Gosset\footrecall{iqc}
	\footrecall{co}
	\and Pawel Wocjan\footrecall{ibm}
}
\date{}
\begin{document}

	\maketitle

	\begin{abstract} 
		The partition function and free energy of a quantum many-body system determine its physical properties in thermal equilibrium.
		Here we study the computational complexity of approximating these quantities for $n$-qubit local Hamiltonians. 
		First, we report a classical algorithm with $\mathrm{poly}(n)$ runtime which approximates the free energy of a given
		$2$-local Hamiltonian provided that it satisfies a certain denseness condition. 
		Our algorithm combines the variational characterization of the free energy and 
		convex relaxation methods. It contributes to a body of work 
		on efficient approximation algorithms for dense instances of optimization problems which are hard in the general case, and can be viewed as simultaneously extending existing algorithms for (a) the ground energy of dense $2$-local Hamiltonians, and (b) the free energy of dense classical Ising models. Secondly, we establish polynomial-time equivalence between the problem of
		approximating the free energy of local Hamiltonians
		and three other natural quantum approximate counting problems, including the problem of
		approximating the number of witness states accepted by a $\mathsf{QMA}$ verifier. These results suggest that simulation of quantum many-body systems in thermal equilibrium may precisely capture the complexity of a broad family of computational problems that has yet to be defined or characterized in terms of known complexity classes. Finally, we summarize state-of-the-art classical and quantum algorithms for approximating the free energy and show how to
		improve their runtime and memory footprint. 
	\end{abstract}


	\section{Introduction}
	\label{sec:intro}
	
	Predicting properties of a quantum many-body system that emerge from 
	its microscopic description in terms of constituent particles and interactions among them
	is a fundamental problem in physics. Many properties of a system in thermal equilibrium are determined by
	the partition function 
	\[
	\zee=\trace[e^{-\beta H}],
	\]
	where $\beta$ is the inverse temperature
	and $H$ is the Hamiltonian describing the system. 
	The partition function appears as a normalization factor
	in the thermal equilibrium  Gibbs state $\rho=e^{-\beta H}/\zee$  and determines important  thermodynamic
	quantities such as the Helmholtz free energy 
	\begin{equation}
		F=-(1/\beta)\log{\zee}.
		\label{eq:fbeta}
	\end{equation}
	The ability to calculate the free energy
	and its derivatives  with respect to the temperature
	and Hamiltonian parameters such as external fields
	is instrumental for 
	mapping out the phase diagram of the system
	and predicting physical properties of each phase.
	Accordingly, the problem of estimating the  free energy
	has been extensively studied, both in the physics and computer science
	communities. Here we investigate this problem for $k$-local Hamiltonians~\cite{kitaev2002book} that describe a system of $n$ qubits with interactions among subsets of at most $k$ qubits. Despite its fundamental significance, little is known about the computational complexity of estimating the free energy of a local Hamiltonian to a given additive error (or equivalently the partition function to a given relative error).   All we can say is that the problem is $\mathsf{QMA}$-hard, even for $2$-local Hamiltonians \cite{kempe2006complexity}, and can be solved in polynomial time using a $\#\mathsf{P}$ oracle~\cite{brown2011computational}.  As we discuss below, precisely characterizing the complexity of this problem  appears to be challenging. Nevertheless, the $\mathsf{QMA}$-hardness result indicates that the best one can hope for is an efficient algorithm for the free energy of some special classes of local Hamiltonians, or which achieves a less stringent approximation.
	
	\paragraph{Efficient approximation algorithm for dense $2$-local Hamiltonians}
	Our first contribution is an efficient classical algorithm that approximates the free energy
	for any $n$-qubit $2$-local Hamiltonian $H=\sum_{i,j} H_{ij}$ satisfying a denseness condition
	\be
	\label{dense_condition}
	\max_{i,j} \|H_{ij}\| \le O(1/n^2) \sum_{i,j} \| H_{ij}\|.
	\ee
	In other words, if one picks a pair of qubits $i,j$ uniformly at random then
	the average interaction strength $\mathbb{E}_{i,j} \|H_{ij}\|$ is  within a constant factor of the maximum
	interaction strength $\max_{i,j} \|H_{ij}\|$. For example, if 
	qubits are located at vertices of a graph such that $\|H_{ij}\|=1$ whenever $i,j$ are connected by an edge and
	$H_{ij}=0$ otherwise then the Hamiltonian is dense iff the graph has $\Omega(n^2)$ edges.
	Our algorithm takes as input
	a dense $2$-local Hamiltonian $H$, the inverse temperature $\beta$,  
	a precision parameter $\epsilon$, and outputs  an estimator $\tilde{F}$ satisfying
	\be
	\label{intro1}
	|F-\tilde{F}| \le \epsilon \sum_{i,j} \|H_{ij}\|.
	\ee
	The algorithm has  runtime $n^{O(1/\epsilon^2)}$ which is polynomial in $n$ for any constant
	$\epsilon>0$.
	Note that the runtime does not depend on the temperature. Moreover, while the above result applies to $2$-local Hamiltonians, we believe that a very similar approach can be used to prove the analogous statement for $k$-local Hamiltonians.
	A  more precise statement of our result can be found in Theorem~\ref{thm:free} in Section~\ref{sec:dense}. 
	
	In most cases of interest the right hand side of Eq.~\eqref{intro1} is expected to increase with $n$ (even for constant $\epsilon$), so our result falls short of providing a small constant additive error approximation of $F$. However, we do not expect that the latter approximation guarantee can be achieved by any polynomial time algorithm.
	Indeed, it can be easily shown that the problem of approximating the free energy with a small additive error
	has the same complexity for dense and for general $2$-local Hamiltonians\footnote{Indeed,
		suppose $H_1$ is a $2$-local Hamiltonian on $n$ qubits and $H_2$
		is a dense $2$-local Hamiltonian on $m$ qubits 
		such that the free energy of $H_2$ is easy to compute (for example, one can choose 
		$H_2=\sum_{1\le i<j\le m} Z_i Z_j$).  Then a Hamiltonian
		$H=H_1\otimes I + I \otimes H_2$ is dense for $m=poly(n)$ and
		the free energy of $H$ is $F=F_1+F_2$,
		where $F_i$ is the free energy of $H_i$. 
	}.
	As noted above, the latter problem is $\mathsf{QMA}$-hard. The assumption that the Hamiltonian is dense also cannot be removed without compromising the quality of approximation or the runtime of the algorithm.
	Indeed, the famous PCP theorem~\cite{dinur2007pcp} implies that approximating the ground energy of general (not necessarily dense) $2$-local Hamiltonians
	with an additive error $\epsilon \sum_{i,j} \|H_{ij}\|$ and a sufficiently
	small constant $\epsilon>0$ is
	$\mathsf{NP}$-hard, see for instance~\cite{alimonti1997hardness}. 
	This is true even for classical (diagonal) Hamiltonians associated with optimization
	problems such as Max-Cut~\cite{alimonti1997hardness}.
	This hardness of approximation carries over to the free energy $F$
	since the latter approximates the ground energy for large $\beta$. In particular, a straightforward argument\footnote{Using the fact that the derivative of the free energy with respect to temperature $T=\beta^{-1}$ is upper bounded by the entropy of the Gibbs state, which is at most $n\log(2)$.} shows that the free energy with $\beta=O(n\delta^{-1})$ attains a $\delta$-additive error approximation to the ground energy.
	
	Our result complements existing classical methods for simulating
	quantum systems in thermal equilibrium. 
	For example, quantum Monte Carlo (QMC)  is a suite of classical
	simulation methods that map the quantum partition function to the one
	describing a classical spin system and approximate the latter using
	Markov Chain Monte Carlo algorithms~\cite{suzuki1993quantum}. 
	However, the quantum-to-classical mapping
	employed by all QMC algorithms is well defined only for Hamiltonians avoiding the so-called sign problem
	(aka stoquastic Hamiltonians)~\cite{bravyi2006complexity}. 
	Furthermore, the efficiency of QMC hinges on 
	the rapid mixing property of the underlying Markov Chain which has been
	rigorously proved only in a few special cases~\cite{crosson2021rapid,bravyi2017polynomial}.
	Another popular technique for estimating the free energy 
	is the high-temperature expansion~\cite{barvinok2016combinatorics,harrow2020classical}.
	It works by computing 
	the Taylor series of the log-partition function 
	$\log{\zee}$ considered as a function of the variable $\beta$ at the point $\beta=0$
	and truncating the series at a sufficiently high order.
	The method is provably efficient assuming that the partition  function  is zero-free
	in a  neighborhood of the interval $[0,\beta]$ in the complex plane~\cite{harrow2020classical}.
	The desired zero-freeness property is expected to hold for all temperatures
	above the phase transition point where  
	a complex zero of the partition function approaches the real axis.
	
	Our result also extends a body of work on
	efficient approximation algorithms for dense instances of 
	classical and quantum optimization
	problems which  are $\mathsf{NP}$- and $\mathsf{QMA}$-hard respectively
	in the general case~\cite{arora1999polynomial,brandao2013product,gharibian2012approximation}. For example, Ref.~\cite{gharibian2012approximation} described an algorithm for dense local Hamiltonians and showed that it gives a good approximation to the minimum energy $\mathrm{Tr}(H\rho)$ achieved by any $n$-qubit product state $\rho=\rho_1\otimes \rho_2\ldots \otimes \rho_n$. Later it was shown \cite{brandao2013product} that the same algorithm gives a good approximation to the true ground energy. Ref. \cite{brandao2013product} also described a different approximation algorithm based on convex relaxation methods for the ground energy of $2$-local Hamiltonians on graphs with low threshold rank (see \cite{brandao2013product} for a definition of this property). As far as we know, the techniques used in these papers for minimizing ground energies cannot be directly applied to the problem of estimating the free energy.
	
	Instead, the starting point for our algorithm is a beautiful paper by Risteski~\cite{risteski2016partition}
	that showed how to approximate the free energy of classical Ising-type Hamiltonians
	satisfying a denseness condition analogous to Eq.~(\ref{dense_condition}).
	Following  Risteski~\cite{risteski2016partition}, we use a variational characterization of the free energy,
	namely 
	\begin{equation}
		F = \min_\rho f(\rho) \quad \text{where}\quad  f(\rho)=\mathrm{Tr}(\rho H)+ (1/\beta) \mathrm{Tr}(\rho \log{\rho})
		\label{eq:varf}
	\end{equation}
	and the minimization is over all $n$-qubit density matrices $\rho$.
	To make this optimization problem tractable, we consider its relaxation defined in terms of all $m$-qubit reduced density matrices satisfying local consistency conditions. The optimal value $\tilde{F}$ of the relaxed problem gives a lower bound on the free energy $F$. Here $m$, which determines the size of the convex program, depends only on the precision $\epsilon$ and a parameter quantifying the denseness of the Hamiltonian. Since we consider 2-local Hamiltonians, the $m$-qubit reduced density matrices suffice to give a relaxation for the average energy term in $f(\rho)$. This idea has been used previously in classical computer science to obtain efficient algorithms for dense instances of constraint satisfaction problems; the relaxed program there is defined in terms of marginals of classical probability distributions~\cite{yoshida2014approximation}. Risteski's insight was that a suitable relaxation to the entropy term in Eq.~\eqref{eq:varf}, a \emph{pseudo-entropy} function, can be computed from these marginals~\cite{risteski2016partition}.

	We define a relaxation to the quantum free energy minimization problem following a similar roadmap. We show that the relaxed problem is convex (though it is not a semidefinite program). Thus its optimal value $\tilde{F}$ can be computed efficiently using known convex optimization methods \cite{bertsimas2004solving,grotschel2012geometric}. The last step of our algorithm is a \emph{rounding map} that converts the optimal solution of the relaxed problem (i.e. a collection of $m$-qubit reduced density matrices) to an $n$-qubit density matrix $\rho$ such that $f(\rho)\le \tilde{F} + \epsilon \sum_{ij} \|H_{ij}\|$. This is the only step of our algorithm that requires the denseness condition Eq.~(\ref{dense_condition}). Combining the inequalities$\tilde{F}\le F\le f(\rho)\le \tilde{F} + \epsilon \sum_{i,j} \|H_{ij}\|$ gives the desired approximation error bound  Eq.~(\ref{intro1}).

	A key technical component is our rounding technique, which has its origins in the propagation sampling method of Ref. \cite{barak2011rounding} and its quantum variant proposed by  Br\~{a}ndao and Harrow~\cite{brandao2013product} to obtain an  
	approximation algorithm for the ground energy of $2$-local Hamiltonians on graphs with low threshold rank. For the considered family of graphs the quantum rounding map proposed in Ref.~\cite{brandao2013product} has the feature that the value of the energy $\mathrm{Tr}(\rho H)$ of the output of the map can be bounded in terms of its value before rounding. Here we require this \textit{energy preservation} property for a different family of graphs, but more importantly we require a version of the quantum rounding map that also plays nicely with the relaxation to the entropy. We are able to make this work by combining technical ingredients from Refs. \cite{risteski2016partition, yoshida2014approximation, brandao2013product}; in addition to the energy preservation property we show that the quantum rounding map has the crucial feature that the von Neumann entropy of its output (a density matrix) is at least as large as the pseudo-entropy function of its input (a collection of $m$-qubit reduced density matrices).  
	
	Our result simultaneously generalizes the efficient algorithm for the ground energy of dense $2$-local Hamiltonians from Refs.\cite{brandao2013product,gharibian2012approximation} (obtained as a special case in which $\beta$ is sufficiently large) and the one for the free energy of dense classical Ising models \cite{risteski2016partition}. It is worth noting, however, that if we consider the special case of ground energies of dense $2$-local Hamiltonians, our algorithm is different and arguably simpler than the one from Refs.~\cite{brandao2013product,gharibian2012approximation}\footnote{The algorithm in those works is based on approximating the minimum ground energy over $n$-qubit product states \cite{gharibian2012approximation} and then uses the fact \cite{brandao2013product} that the optimal product state energy is a good approximation to the true ground energy.}.

	\paragraph{Complexity of partition functions}
	Next we consider the computational complexity of approximating the free energy $F$
	for more  general Hamiltonians
	describing a system of $n$ qubits with $k$-qubit interactions, where $k=O(1)$.
	Such a system can be described by a $k$-local Hamiltonian 
	\[
	H=\sum_{S\subseteq \{1,\ldots,n\}}\; H_S
	\]
	where each term $H_S$ is a hermitian operator acting non-trivially only on the subset of qubits $S$
	and $H_S=0$ unless $|S|\le k$.
	Below we assume that $\|H_S\| \le \mathrm{poly}(n)$ for all $S$. 
	Since a small additive error approximation of the free energy  is equivalent to a small relative error approximation
	of the partition function, we shall focus on the latter.
	Thus we are interested in the following problem. 
	
	\begin{problem*}[\bf Quantum Partition Function (QPF)]
		Given a  $k$-local Hamiltonian $H$ acting on $n$ qubits,
		inverse temperature $\beta\le \mathrm{poly}(n)$, 
		and a precision parameter $\delta\ge \mathrm{poly}(1/n)$. Compute an estimate $\tilde{Z}$ such that 
		\[
		(1-\delta) \zee  \le \tilde{Z} \le (1+\delta) \zee.
		\]
		Here $\zee=\mathrm{Tr}(e^{-\beta H})$.
	\end{problem*}
	We note that it is crucial here that we consider \textit{relative error} estimation of the partition function. Previous works have considered the less challenging problem of additively approximating the (normalized) partition function \cite{brandao08,cade2017quantum,chowdhury2020computing}. These works show that the latter problem admits an efficient quantum algorithm and, for Hamiltonians with locality $k=O(\log(n))$, is complete for the complexity class $\mathsf{DQC}_1$ of problems that can be solved in polynomial time with only ``one clean qubit" \cite{knill1998power}. 
	
	In contrast to the ground energy, the partition function $\zee$ depends on all eigenvalues
	of the Hamiltonian as well as their \textit{degeneracy}. Thus
	one can view the QPF problem as a quantum analogue of 
	approximate counting — a fundamental task in classical computational complexity.
	In particular, it is known that the problem of approximating the partition function of a classical local
	Hamiltonian with a small relative error is 
	contained in the complexity class $\mathsf{BPP}^{\mathsf{NP}}$, essentially due to a fundamental result of Stockmeyer~\cite{stockmeyer1983complexity}. We can interpret this result as stating that the 
	classical version of the QPF problem  is not much harder than $\mathsf{NP}$,
	which is surprising because computing the partition function exactly is $\#\mathsf{P}$-hard. 
	On the other hand, if we allow $\beta$ to be complex, the 
	task of approximating $\zee$
	with a constant relative error is known to be $\#\mathsf{P}$-hard~\cite{goldberg2017complexity},
	even in the case of classical Hamiltonians. The discrepancy between the problem complexity for real and complex $\beta$’s suggests that non-negativity of the Gibbs state should play an important role in the complexity analysis.

	The fact that Stockmeyer's approximate counting method does not generalize to the quantum case has been observed previously \cite{aharonov2008pursuit}.
	This may explain why the complexity of the QPF problem
	-- a basic and natural
	question from a physics perspective, remains largely open.
	A recent result by Cubitt et al.~\cite{cubitt2018universal} gives us hope that this question may have
	a mathematically appealing answer. The work~\cite{cubitt2018universal}  introduced the notion of universal quantum Hamiltonians
	and showed (among other things) that the QPF problem for
	$k$-local Hamiltonian has  the same complexity for any constant $k\ge 2$. (In fact, one can further specialize $2$-local Hamiltonians to the anti-ferromagnetic Heisenberg
	model on a two-dimensional lattice of qubits or even certain 1D qudit systems~\cite{cubitt2018universal, zhou2021strongly}.)
	The same universality result also applies to natural families of fermionic Hamiltonians
	relevant for quantum chemistry and material science~\cite{cubitt2018universal}.
	This suggests that the  QPF problem is complete
	for some ``universal" complexity class that can be
	defined in terms of a suitable computational model and does not depend
	on details of the considered Hamiltonians, such as the locality parameter $k$
	(as long as $k\ge 2$)
	or whether the underlying system consists of qubits (spins) or fermions.
	
	Given this state of affairs, and the apparently challenging nature of characterizing the complexity of the QPF
	problem in terms of known complexity classes, here we use a standard computer science dodge -- we look for other computational problems which are equivalent  to QPF under polynomial time reductions
	with the hope that these other problems may ultimately be easier to understand.
	Namely, we show that the QPF problem is equivalent to approximating the following quantities:
	\begin{enumerate}
		\item[(1)] The expected value of a local observable 
		in the Gibbs state of a $k$-local Hamiltonian.
		\item[(2)] The number of eigenvalues of a $k$-local Hamiltonian $H$ within a given interval.
		\item[(3)] The number of witness states accepted by a $\mathsf{QMA}$ verifier\footnote{Here $\mathsf{QMA}$ verifier is a polynomial-size quantum circuit
			followed by a measurement of some designated output qubit. The circuit
			takes as input a witness state and possibly ancilla qubits
			initialized in $|0\rangle$. A witness state is accepted if the 
			probability of the measurement outcome `1' is above a specified threshold.}
		
	\end{enumerate}
	We require an approximation within a small additive error in case (1) and
	a small relative error in cases (2,3). For a formal definition of the above problems
	see Section~\ref{sec:equivalence}. 
	Problems~(1,2) are ubiquitous in condensed matter physics
	since expected values of local observables and the density of states
	(i.e. the number of eigenstates contained in a given energy interval)
	provide important insights into properties of a quantum material.
	Problem~(3) can be viewed as a natural quantum generalization of the approximate counting task (the exact version of this problem introduced by Brown et al~\cite{brown2011computational}
	is known to be $\#\mathsf{P}$-hard).
	We also reproduce the result of Cubitt et al.~\cite{cubitt2018universal}  showing that
	the QPF problems for $2$-local and $k$-local Hamiltonians with any
	constant $k$ are equivalent to each other. Our proof is slightly more direct than the one of~\cite{cubitt2018universal} as we do not use perturbation theory gadgets; our technique appears to share some features of later works which used Kitaev's circuit-to-Hamiltonian mapping \cite{zhou2021strongly, kohler2020translationally, kohler2021general}. Note however that reducing $k$-local partition functions to 2-local partition functions is simpler than showing that such Hamiltonians are universal in the sense considered in Refs. \cite{cubitt2018universal, zhou2021strongly, kohler2020translationally, kohler2021general}.
	
	\paragraph{Improved exponential-time algorithms via Clifford compression}
	We conclude the paper by revisiting exponential-time classical and quantum
	algorithms solving the QPF problem for general local Hamiltonians. We show that a \textit{Clifford compression technique}~\cite{gosset2019compressed, huang2020predicting} based on the unitary 2-design property of the 
	Clifford group~\cite{cleve2015near,dankert2009exact} can be used to improve the runtime of state-of-the art classical algorithms for QPF. We also show how the Clifford compression technique can be used to almost halve the memory footprint of state-of-the-art quantum algorithms for QPF without compromising their runtime.
	
	First, we give a classical algorithm that solves the QPF problem for any $n$-qubit
	$k$-local Hamiltonian $H$ 
	with $\beta \|H\|\le b$ in time  
	\be
	\label{intro2}
	O\left(
	(b + \log{(1/\delta)})\ell 2^n\delta^{-1}  + n^2 2^n\delta^{-1} + \delta^{-4}
	\right),
	\ee
	where $\ell\le O(n^k)$ is the number of non-zero $k$-local terms that appear in $H$.
	For example, if $H$ is a geometrically local Hamiltonian on a regular lattice with bounded strength interactions
	then $\|H\|=O(n)$ and $\ell=O(n)$.
	For such Hamiltonians the runtime in Eq.~(\ref{intro2})
	becomes $O((1+\beta) n^2 2^n/\delta)$, ignoring
	terms logarithmic in $1/\delta$ and assuming $2^n\ge 1/\delta^3$. Our algorithm is based on the so-called
	Hutchinson’s estimator~\cite{hutchinson1989stochastic}  and its improved version
	known as Hutch++~\cite{meyer2021hutch++}.
	Namely, suppose $A$ is a positive semidefinite real matrix of size $d\times d$
	such that a matrix-vector product $A|\psi\ra$ can be computed
	in time $t(A)$ for any input vector $|\psi\ra\in \RR^d$.
	Hutch++~\cite{meyer2021hutch++} is a classical algorithm that approximates
	$\mathrm{Tr}(A)$ within a relative error $\delta$ 
	and has runtime $O(t(A)/\delta + d/\delta^2)$.
	This algorithm is asymptotically optimal in the query complexity model~\cite{meyer2021hutch++}. 
	We report a modified version 
	of Hutch++ which has runtime
	$O(t(A)/\delta+d\log^2{(d)}/\delta)$, assuming that $d\ge 1/\delta^3$.
	This improves the runtime scaling
	with  $\delta$ almost quadratically.
	Our construction applies the original Hutch++ algorithm to a certain compressed
	version of the matrix $A$ obtained by projecting $A$ onto  the logical subspace
	of a random stabilizer code with $O(\log{1/\delta})$ logical qubits.
	It shares some features of the methods used in Refs.~\cite{gosset2019compressed, huang2020predicting}  to obtain compressed classical descriptions of quantum states. 
	We emphasize that the improved Hutch++
	algorithm is purely classical, even though its analysis relies on quantum techniques.
	Specializing the improved Hutch++ to the case $A=\mathrm{Re}(e^{-\beta H})$ with a $k$-local Hamiltonian $H$
	gives a classical algorithm for the  QPF problem with runtime
	Eq.~(\ref{intro2}), see Section~\ref{sec:hutch} for details. 
	
	Secondly, we give a memory-efficient quantum algorithm for the QPF problem which requires only $O(\log(n)+\log(1/\delta))$ ancilla qubits. This is an improvement in space requirements over  an earlier algorithm due to Poulin and Wocjan~\cite{poulin09sampling} which needed an $\Omega(n)$-sized ancilla register. 
	The original algorithm constructs a quantum circuit to prepare a purification of a finite-temperature Gibbs state using Hamiltonian simulation and quantum phase estimation. There, the normalized partition function $\zee/2^n$ becomes the probability of post-selecting the purified Gibbs state and can be determined with quantum amplitude estimation~\cite{brassard2002amplitude}. The overhead of $\Omega(n)$ ancilla qubits in this approach is a consequence of purifying an $n$-qubit mixed state, and cannot be reduced using, e.g., modern methods for eigenvalue transformations~\cite{chowdhury2017quantum,gilyen2019singular,vanApeldoorn2020quantumsdpsolvers}.  
	
	We use the unitary 2-design construction of Section~\ref{sec:hutch} to circumvent the purification step. Instead, we approximate the normalized partition function as a transition probability in a randomly chosen ``compressed'' quantum circuit which acts on fewer qubits. The running-time of our algorithm is $\tilde O(\sqrt{2^n/\zee}\cdot \beta/\delta)$, comparable to the running-time obtained by improving the algorithm of Ref.~\cite{poulin09sampling} with newer techniques~\cite{chowdhury2017quantum,gilyen2019singular,vanApeldoorn2020quantumsdpsolvers}. 
	For technical reasons, the quantum algorithm assumes that each local term in the Hamiltonian is also positive semi-definite but note that this is not critical for the savings in ancilla.
	
	The rest of this paper is organized as follows. We describe our approximation algorithm for the free energy of dense $2$-local Hamiltonians in Section~\ref{sec:dense}. We establish polynomial time reductions between the QPF problem and three other quantum approximate counting problems in Section \ref{sec:QPF}. We review the Hutch++ algorithm and describe the Clifford compression technique that leads to its improved version in Section~\ref{sec:hutch}. Finally, in Section~\ref{sec:quantum} we show how Clifford compression can be incorporated into a quantum algorithm for the QPF problem leading to a reduced memory footprint.

	\section{Approximation algorithm for dense $2$-local Hamiltonians}
	\label{sec:dense}
	In this Section we describe our efficient classical approximation algorithm for the free energy of dense $2$-local Hamiltonians.
	
	Let us consider an $n$-qubit Hamiltonian
	\[
	H=\sum_{i,j=1}^{n} H_{ij}
	\] 
	where each local term $H_{ij}$ acts nontrivially only on qubits $i,j$. Let $\Gamma_{ij}=\|H\|_{ij}$ and
	\[
	\Gamma=\sum_{i,j} \Gamma_{ij}
	\]
	We assume without loss of generality that $\Gamma_{ii}=0$ for all $i$. We say that the Hamiltonian $H$ is $\Delta$-dense if 
	\[
	\Gamma \geq \Delta n^2 \max_{i,j} \Gamma_{ij}.
	\]

	We shall be interested in the free energy Eq.~\eqref{eq:fbeta}. For ease of notation and without loss of generality in this Section we consider the case $\beta=1$ and
	\begin{equation}
		F\equiv-\log(\mathrm{Tr}(e^{-H})).
		\label{eq:fone}
	\end{equation}
	Here and below we use the natural logarithm. The free energy at any other value $\beta\neq 1$ can be obtained by computing Eq.~\eqref{eq:fone} with the substitution $H\leftarrow \beta H$ in Eq.~\eqref{eq:fone} and then dividing by $\beta$, i.e. $F\leftarrow F/\beta$. Using this procedure to estimate the free energy Eq.~\eqref{eq:fbeta} for any value of $\beta$ we obtain the same runtime and approximation guarantee stated in Theorem \ref{thm:free}\footnote{To see this, note that $\Gamma_{ij}\leftarrow \beta \Gamma_{ij}$ and $\Gamma\leftarrow \beta \Gamma$ as well.}.
	
	Risteski \cite{risteski2016partition} described a nontrivial approximation algorithm for dense \textit{classical} $2$-local Hamiltonians. Here we establish the following
	quantum generalization of Risteski's result. 
	
	\begin{theorem}
		Suppose $H$ is a $\Delta$-dense $2$-local $n$-qubit Hamiltonian and let $\epsilon\in (0,1)$ be given. Suppose $n\geq \epsilon^{-3}\Delta^{-2}$. There is a classical algorithm with runtime $n^{O(\Delta^{-1}\epsilon^{-2})}$ which computes an estimate $\tilde{F}$ such that $|F-\tilde{F}|\leq \epsilon \Gamma$.
		\label{thm:free}
	\end{theorem}
	
	The proof of Theorem \ref{thm:free}, following the strategy from Ref. \cite{risteski2016partition}, begins with a variational characterization of the free energy which is given below. This variational characterization involves optimizing over the set of $n$-qubit density matrices, which is too costly for our purposes. We then describe a relaxation of the free energy which involves optimizing over a convex set of \textit{pseudodensity matrices} of comparatively much smaller size. The optimal value of the convex relaxation can be computed efficiently on a classical computer using known techniques \cite{bertsimas2004solving}. Finally, we show that the optimal value of the relaxation is close to the true free energy. This step uses a rounding which maps pseudodensity matrices to density matrices in a suitable way.

	Let us write $\mathcal{B}$ for the set of $n$-qubit density matrices. For any $\sigma\in \mathcal{B}$, define $S(\rho)=-\mathrm{Tr}(\rho\log(\rho))$ to be its von Neumann entropy. Our starting point is the following expression for the free energy.

	\begin{lemma}[Variational characterization of free energy]
		\begin{equation}
			F=\min_{\rho\in \mathcal{B}} f(\rho) \qquad \text{where} \qquad f(\rho)\equiv  \mathrm{Tr}(H\rho)-S(\rho).
			\label{eq:varchar}
		\end{equation}
	\end{lemma}
	\begin{proof}
		The lemma is a simple consequence of the fact \cite{nielsenchuang} that for any two density matrices $\rho, \sigma$, the quantum relative entropy is nonnegative, i.e., 
		\[
		S(\rho)\leq -\mathrm{Tr}(\rho\log(\sigma)).
		\]
		Applying this with $\sigma=e^{-H}/\mathrm{Tr}(e^{-H})$ we see that for any density matrix $\rho$ we have
		\[
		\mathrm{Tr}(H\rho)-S(\rho)\geq F=f(\sigma).
		\]
	\end{proof}
	
	We are interested in algorithms which run in time $\mathrm{poly}(n)$ so we cannot afford to even store a single $n$-qubit density matrix, which means that directly computing Eq.~\eqref{eq:varchar} is off the table. As noted above, we will consider a relaxation of this optimization problem which involves optimizing over a data structure with smaller size. In particular, a $k$-local \textit{pseudodensity matrix} $\rho$ is a collection of density matrices on subsets $S$ of qudits of size at most $k$, 
	\begin{equation}
		\rho=\{\rho_S: S\subseteq [n], |S|\leq k\}.
		\label{eq:ps1}
	\end{equation}
	Each  $\rho_S$ must be a valid $|S|$-qubit density matrix, i.e., 
	\begin{equation}
		\rho_S\in \mathbb{C}^{2^{|S|}\times 2^{|S|}} \quad \text{and} \quad \mathrm{Tr}(\rho_S)=1 \quad \text{ and } \quad  \rho_S\geq 0.
		\label{eq:ps2}
	\end{equation}
	Moreover, the collection of density matrices must be consistent in the sense that for all subsets $C,D$ of size at most $k$ we have
	\begin{equation}
		\mathrm{Tr}_{C\setminus D} (\rho_C) =\mathrm{Tr}_{D\setminus C}(\rho_D).
		\label{eq:ps3}
	\end{equation}
	We write $\mathcal{B}_k$ for the set of $k$-local, $n$-qubit pseudodensity matrices.

	Note that the function $f(\rho)$ is well defined for density matrices $\rho\in \mathcal{B}$ and we would like a substitute which takes pseudodensity matrix input. To this end we define a quantum pseudo-entropy function $S_k: \mathcal{B}_k\rightarrow \mathbb{R}$ as follows:
	\begin{equation}\label{eq:pseudoS}
		S_{k}(\sigma) \equiv \min_{C\subseteq [n], |C|<k} \{S(C)_{\sigma} + \sum_{j\notin C }S(j|C)_{\sigma}\},
	\end{equation}
	where $S(C)_\sigma\equiv S(\sigma_C)$ and
	\begin{equation}\label{eqn:quant_cond_entropy}
		S(j |C)_{\sigma} \equiv S(C \cup \{j\})_{\sigma} - S(C)_{\sigma}.
	\end{equation}
	Note that these entropies are all well defined since $|C|\leq k-1$ and $\sigma\in \mathcal{B}_k$.

	By analogy with Eq.~\eqref{eq:varchar} we define $f_k:\mathcal{B}_k\rightarrow \mathbb{R}$ as
	\[
	f_k(\sigma)=\sum_{i,j}\mathrm{Tr}(H_{ij}\sigma_{ij})- S_k(\sigma).
	\]
	
	Let us define 
	\begin{equation}
		f^{\star}_k\equiv \min_{\sigma\in \mathcal{B}_k} f_k(\sigma).
		\label{eq:fstar}
	\end{equation}

	\begin{lemma}
		The optimal value $f_k^{\star}$ can be computed on a classical computer with runtime $n^{O(k)}$.
		\label{lem:optim}
	\end{lemma}
	Lemma \ref{lem:optim} follows from the fact that $\mathcal{B}_k$ is a convex set and for each $C\subseteq [n]$ of size $<k$ the function $S(C)_{\sigma} + \sum_{j\notin C }S(j|C)_{\sigma}$ is a concave function of the pseudodensity matrix (represented suitably as a real vector, see Section \ref{sec:optimize} for details). Therefore $S_k(\sigma)$, a minimum of concave functions, is concave, and the function $f_k(\sigma)$ is convex.  We may then compute the minimum $f_k^{\star}$ over $\sigma$ using efficient convex optimization algorithms such as the one from Ref.~\cite{bertsimas2004solving} which takes time $n^{O(k)}$ as we discuss in Section \ref{sec:optimize}. 
	
	It remains to show that $f_k^{\star}$ is a good approximation to the free energy $F$ for suitably chosen $k$. To show this we will use the following two lemmas. 
	
	\begin{lemma}
		Suppose $\rho\in \mathcal{B}$ is an $n$-qubit density matrix. Then $S_k(\rho)\geq S(\rho)$.
		\label{lem:relaxation}
	\end{lemma}
	\begin{proof}
		The definition of the quantum conditional entropy for any density matrix $\rho$ gives
		\begin{equation}
			S(\rho) = \sum_{j=1}^n S(j|{1,\dots,j-1})_\rho \leq \sum_{j=1}^n S(\rho_j)\;,
		\end{equation}
		the second step following from the non-negativity of the quantum mutual information. 
		
		Now let $C\subseteq [n]$ be such that $S_k(\rho)=S(C)_\rho+\sum_{j\notin C}S(j|C)_\rho$. For ease of notation let us suppose that $C=\{1,2,\ldots, |C|\}$. Similarly to the above argument we may write 
		\begin{align}
			S(\rho) &= S(C)_\rho + \sum_{j=|C|+1}^{n} S(j|{C\cup\{j-1,j-2,\dots ,|C|+1})_\rho \\
			&\le S(C)_{\rho} + \sum_{j=|C|+1}^{n} S(j|C)_\rho\\
			&=S_k(\rho).
		\end{align}
		the second step now following from the strong subadditivity of the quantum mutual information~\cite{watrous_2018}. 
	\end{proof}
	\begin{lemma}
		There exists a \textbf{rounding} which maps any pseudodensity matrix $\sigma\in \mathcal{B}_{k}$ to a density matrix $\rho\in \mathcal{B}$ such that
		\begin{equation}
			S(\rho)\geq S_k(\sigma).
			\label{eq:entrop}
		\end{equation}
		Furthermore, if $n\geq \epsilon^{-3}\Delta^{-2}$ then we may choose $k=O(\Delta^{-1} \epsilon^{-2})$ such that
		\begin{equation}
			\left|\sum_{i,j} \mathrm{Tr}(H_{ij}\sigma_{ij})-\mathrm{Tr}(\rho H)\right|\leq \epsilon \Gamma.
			\label{eq:energ}
		\end{equation}
		\label{lem:round}
	\end{lemma}
	The proof of Lemma \ref{lem:round} is provided in Section \ref{sec:round}
	\begin{proof}[\textbf{Proof of Theorem \ref{thm:free}}]
		Let us choose $k=O(\Delta^{-1}\epsilon^{-2})$ as prescribed in Lemma \ref{lem:round} so that Eqs.~(\ref{eq:entrop}, \ref{eq:energ}) hold. With this choice our estimate of the free energy is $f^{\star}_k$. Lemma \ref{lem:optim} states that we can compute $f^{\star}_k$ using a runtime $n^{O(\Delta^{-1}\epsilon^{-2})}$. Below we show that 
		\begin{equation}
			0\leq F-f^{\star}_k\leq \epsilon \Gamma.
			\label{eq:uplower}
		\end{equation}
		To this end, let $\sigma\in \mathcal{B}_k$ be such that $f^{\star}_k=f_k(\sigma)$, and let $\rho$ be the image of $\sigma$ under the rounding map described in Lemma \ref{lem:round}. We have
		\[
		F-f^{\star}_k=F-f_k(\sigma)\leq f(\rho)-f_k(\sigma)
		\label{eq:ffstar1}
		\]
		where we used the fact that $F\leq f(\rho)$ for all $\rho\in \mathcal{B}$. Substituting the definitions in the above gives
		\begin{align}
			F-f^{\star}_k&\leq \mathrm{Tr}(H\rho)-\sum_{i,j} \mathrm{Tr}(H_{ij} \sigma_{ij}) +S_k(\sigma)-S(\rho)\\
			& \leq \mathrm{Tr}(H\rho)-\sum_{i,j} \mathrm{Tr}(H_{ij} \sigma_{ij})\\
			&\leq \epsilon \Gamma
		\end{align}
		where in the second line we used Eq.~\eqref{eq:entrop} and in the last line we used Eq.~\eqref{eq:energ}. This establishes the upper bound in Eq.~\eqref{eq:uplower}. For the lower bound we use Lemma \ref{lem:relaxation} from which we directly see that $f(\rho)\geq f_k(\rho)$ for any $\rho\in \mathcal{B}$. Therefore
		\[
		F=\min_{\rho\in \mathcal{B}} f(\rho)\geq \min_{\rho\in \mathcal{B}} f_k(\rho)\geq \min_{\sigma\in \mathcal{B}_k} f_k(\sigma)=f^{\star}_k,
		\]
		which completes the proof. In the above we used the fact that the marginals on $\leq k$ qubits of any n-qubit density matrix $\rho$ define a valid pseudodensity matrix $\sigma \in \mathcal{B}_k$, and therefore the minimum of $f_k$ over the latter set can only be smaller than the minimum over $\mathcal{B}$.
		
	\end{proof}
	
	\subsection{The rounding map \label{sec:round}}
	In this Section we define the rounding map from Lemma \ref{lem:round} and prove the lemma. Let $P_1=X,P_2=Y,P_3=Z$ be the single-qubit Pauli operators. For $b=\{1,2,3\}$ let $|\psi_{b,0}\rangle, |\psi_{b,1}\rangle$ denote eigenstates of $P_b$ with eigenvalues $1$ and $-1$ respectively. For $x=(b,r)$ with $b=b_1b_2\ldots b_m \in \{1,2,3\}^m$ and $r=r_1r_2\ldots r_m \in \{0,1\}^m$ we write
	\[
	|\psi_{x}\rangle=|\psi_{b_1,r_1}\rangle\otimes |\psi_{b_2,r_2}\rangle\otimes\ldots \otimes |\psi_{b_m,r_m}\rangle. \qquad x=(b,r)\in \{1,2,3\}^m\times \{0,1\}^m.
	\]
	Let us write $\mathcal{I}(m)=\{1,2,3\}^m\times \{0,1\}^m$ for the index set.
	
	To round a pseudodensity matrix $\sigma\in \mathcal{B}_k$ to a density matrix $\rho\in \mathcal{B}$ , we use a modification of the procedure used by Br\~{a}ndao and Harrow in Ref.~\cite{brandao2013product}. In particular, 
	\begin{align}\label{eq:roundrho}
		\rho = \mathbb{E}_{0\leq m< k}\Ex_{C\subseteq [n],|C|=m} \sum_{x\in \mathcal{I}(m)} p_C(x) \ketbra{\psi_{x}}_C\otimes \prod_{i\notin C} \rho_i^{(x)}\;.
	\end{align}
	Here each expectation is taken with respect to the uniform distribution, i.e., $\mathbb{E}_{0\leq m<k}(\cdot)=\frac{1}{k}\sum_{m=0}^{k} (\cdot)$. Moreover, $p_C$ is a probability distribution and $\rho_i^{(x)}$ is a single-qubit state defined by
	\begin{equation}
		p_C(x)=3^{-m}\langle \psi_x|\sigma_C|\psi_x\rangle \qquad \qquad \rho_i^{(x)} = \frac{{\mathrm Tr}_C((\ketbra{\psi_{x}}_C\otimes I)\sigma_{C\cup\{i\}})}{3^m p_C(x)}.
		\label{eq:pc}
	\end{equation}
	Note that since $\sigma\in\mathcal{B}_k$ and $|C|=m<k$, the density matrix $\sigma_{C\cup\{i\}}$ appearing above is well defined.
	
	\paragraph{Proof of Eq.~\eqref{eq:entrop}}
	For a fixed $m\in \{0,\ldots, k-1\}$ and $C\subseteq [n]$ of size $|C|=m$ let
	\begin{equation}
		\eta=\sum_{x\in \mathcal{I}(m)} p_C(x) \ketbra{\psi_{x}}_C\otimes \prod_{i\notin C} \rho_i^{(x)}
		\label{eq:rhotil}
	\end{equation}
	be the state associated with $m,C$ in Eq.~\eqref{eq:roundrho}. For ease of notation we suppress the dependence of $\eta$ on $m,C$. Below we show that 
	\begin{equation}
		S(\eta)\geq S_k(\sigma).
		\label{eq:etasigma}
	\end{equation}
	Eq.~\eqref{eq:entrop} then follows using Eq.~\eqref{eq:roundrho} and concavity of the Von Neumann entropy.
	
	\begin{lemma}
		\label{lemma:in}
		We have 
		\begin{equation}
			S(C)_{\eta}\geq S(C)_\sigma.
			\label{eq:sc}
		\end{equation}
		Ordering the qubits so that $C=\{1,2,\ldots, m\}$, we have, for each $i\geq m+1$,
		\begin{equation}
			S(i|\{1,2,\ldots, i-1\})_{\eta}\geq S(i|C)_{\sigma}.
			\label{eq:sicineq}
		\end{equation}
	\end{lemma}
	\begin{proof}
		We first show Eq.~\eqref{eq:sicineq}. Using Eq.~\eqref{eq:rhotil} and writing $x=(b,r)$, we see that for $i\geq m+1$  the reduced density matrix of $\eta$ on qubits $\{1,\ldots, i\}$ is
		\begin{equation}
			\sum_{(b,r)\in \mathcal{I}(m)} p_C(b,r)|\psi_{b,r}\rangle\langle \psi_{b,r}|\otimes  \prod_{m< j\leq i}\rho_j^{(b,r)}=\mathbb{E}_{b} \left[\eta_b\right]
			\label{eq:qop}
		\end{equation}
		where 
		\begin{equation}
			\label{eq:kappab}
			\eta_b=3^m\sum_{r} p_C(b,r)|\psi_{b,r}\rangle\langle \psi_{b,r}|\otimes  \prod_{k\leq j\leq i}\rho_j^{(b,r)}.
		\end{equation}
		Using Eq.~\eqref{eq:qop} and concavity of the conditional entropy (cf.Ref.\cite{nielsenchuang} Corollary 11.13) we have
		\begin{equation}
			S(i|\{1,2,\ldots, i-1\})_\eta\geq \mathbb{E}_b\left[S(i|\{1,2,\ldots, i-1\})_{\eta_b}\right].
			\label{eq:seb}
		\end{equation}
		For any $b$, the state $\eta_b$ has the special feature that 
		\begin{claim} For each $i\geq m+1$ we have
			\label{claim:markov}
			\[
			S(i|\{1,2,\ldots, i-1\})_{\eta_b}=S(i|C)_{\eta_b}.
			\]
		\end{claim}
		\begin{proof}
			The claim states that $\eta_b$ achieves equality in the strong subadditivity inequality $S(Z|X,Y)\leq S(Z|X)$ with respect to subsystems $Z=\{i\}, X=C, Y=\{m+1,\ldots, i-1\}$. To see this note from Eq.~\eqref{eq:kappab} that the reduced density matrix of $\eta_b$ on subsystems $XYZ$ has the form
			\begin{equation}
				(\eta_b)_{XYZ}=\sum_{r} q_r c^r_X\otimes d^r_{Y}\otimes e^{r}_Z
				\label{eq:formofa}
			\end{equation}
			where $\{q_r\}$ is a probability distribution, and for each $r$,  $c^r,d^r,e^r$ are quantum states of subsystems $X,Y$ and $Z$ respectively. Moreover, $c^r$ and $c^{r'}$ have support on orthogonal subspaces whenever $r\neq r'$. Using the form Eq.~\eqref{eq:formofa} of our density matrix and standard facts about Von Neumann entropy (e.g., Eq. 11.57 of Ref.~\cite{nielsenchuang}) we can explicitly compute
			\[
			S(Z|XY)_{\eta_b}=\sum_{r} q_r S(e^r_Z)=S(Z|X)_{\eta_b}.
			\]
		\end{proof}
		Using Claim \ref{claim:markov} in Eq.~\eqref{eq:seb} we get
		\begin{equation}
			S(i|\{1,2,\ldots, i-1\})_{\eta_b}\geq \mathbb{E}_b\left[S(i|C)_{\eta_b}\right].
			\label{eq:sib}
		\end{equation}
		To complete the proof of Eq.~\eqref{eq:sicineq} we now show that $S(i|C)_{\eta_b}\geq S(i|C)_\sigma$ for all $b$. This follows from the fact that the reduced density matrix of $\eta_b$ on $C\cup\{i\}$ is
		\begin{align}
			(\eta_b)_{C\cup\{i\}}&= \sum_{r} 3^m p_C(b,r)|\psi_{b,r}\rangle\langle \psi_{b,r}|\otimes  \rho_i^{(b,r)}\\
			&=\sum_{r} |\psi_{b,r}\rangle\langle \psi_{b,r}| \sigma_{C\cup\{i\}}|\psi_{b,r}\rangle\langle \psi_{b,r}|\label{eq:measb}\\
			&=\mathcal{E}_b\otimes I_i(\sigma_{C\cup\{i\}})
		\end{align}
		where $\mathcal{E}_b$ is the CPTP map that describes measurement in the basis described by $b$. Since CPTP maps do not increase mutual information (cf. Thm 11.15 of Ref.~\cite{nielsenchuang}), Eq.~\eqref{eq:qop} implies
		\begin{equation}
			S(i)_{\sigma}-S(i|C)_{\sigma}\geq S(i)_{\eta_b}-S(i|C)_{\eta_b}.
			\label{eq:mi}
		\end{equation}
		Looking at Eq.~\eqref{eq:measb} we see that the reduced density matrix of  $\eta_b$ on qubit $i$ is equal to $\sigma_i$. Therefore $S(i)_{\sigma}=S(i)_{\eta_b}$ and Eq.~\eqref{eq:mi} implies
		\begin{equation}
			S(i|C)_{\eta_b}\geq S(i|C)_{\sigma}.
			\label{eq:sic}
		\end{equation}
		Plugging Eq.~\eqref{eq:sic} into Eq.~\eqref{eq:sib} gives  Eq.~\eqref{eq:sicineq}.
		
		To show Eq.~\eqref{eq:sc} note that the reduced density matrix of $\eta$ on subsystem $C$ is given by
		\begin{equation}
			\eta_C=3^{-m}\sum_{b} \mathcal{E}_b (\sigma_C)
			\label{eq:eb}
		\end{equation}
		where as before, $\mathcal{E}_b$ is the CPTP map corresponding to a measurement in the orthonormal basis defined by $b$. Using Eq.~\eqref{eq:eb} and fact that Von Neumann entropy is concave, we get
		\[
		S(C)_\eta \geq 3^{-m}\sum_{b} S(\mathcal{E}_b (\sigma_C)).
		\]
		Now we have $S(\mathcal{E}_b (\sigma_C))\geq S(\sigma_C)$ for each $b$ since $\mathcal{E}_b$ descibes measurement in a complete orthonormal basis (cf. Theorem 11.9 of Ref.~\cite{nielsenchuang}). Plugging into the above we arrive at $S(C)_{\eta}\geq S_\sigma(C)$.
	\end{proof}
	Finally, we establish Eq.~\eqref{eq:etasigma}, completing the proof of Eq.~\eqref{eq:entrop}.
	\begin{corol}
		\begin{equation}
			S(\eta)\geq S(C)_\sigma+\sum_{i\notin C} S (i|C)_{\sigma}\geq S_k(\sigma)
			\label{eq:skrho}
		\end{equation}
		
	\end{corol}
	\begin{proof}
		Let us order the qubits so that $C=\{1,\ldots, m\}$. We have
		\[
		S(\eta)=S(C)_\eta+\sum_{j=m+1}^{n} S(j|\{1,2,\ldots, j-1\})_{\eta}\\
		\]
		Substituting Eqs.(\ref{eq:sc}, \ref{eq:sicineq}) into the above completes the proof of the first (leftmost) inequality in Eq.~\eqref{eq:skrho}. The second inequality follows from the definition Eq.~\eqref{eq:pseudoS} and the fact that $|C|=m<k$.
		
	\end{proof}
	\paragraph{Proof of Eq.~\eqref{eq:energ}}
	To establish the bound Eq.~\eqref{eq:energ} on the energy difference before and after rounding, we follow a strategy from Ref. \cite{brandao2013product}. In particular, we relate the pseudodensity matrix $\sigma$ and its rounding $\rho$ to certain classical \textit{pseudodistributions} (a classical analogue of pseudodensity matrices, defined below), and the difference in their energies to a certain expression involving the 1-norm distances between these pseudodistributions.  To bound this expression involving 1-norm distances, we use the following lemma which is a direct consequence of Corollary 1 from Ref. \cite{yoshida2014approximation}. For completeness we include a proof below.
	
	A $t$-local pseudodistribution $p$ over $n$ $d$-ary variables is a collection $\{p_S\}$ of probability distributions corresponding to all subsets $S\subseteq [n]$ of size $|S|\leq t$. Here $p_S:\{1,2,\ldots, d\}^{|S|}\rightarrow \mathbb{R}$ is a  probability distribution over variables which each can take $d$ values. The distributions $p_S$ and $p_T$ are required to have consistent marginals on variables in $S\cap T$. 
	
	\begin{lemma}[\cite{yoshida2014approximation}]
		Suppose $\omega:[n]\times [n]\rightarrow \mathbb{R}$ is a probability distribution such that $\omega(i,i)=0$ for all $i$ and $\Delta n^2\omega(i,j)\leq 1$ for all $i,j\in [n]$. Let $p$ be any $t$-local  pseudodistribution over $n$ $d$-ary variables. Then for any $k<t$ we have 
		\begin{equation}\label{eq:1norm}
			\frac{1}{k}\sum_{0\leq m< k} \; \Ex_{\substack{C\subseteq [n]\\|C|=m}}\; \Ex_{(i,j)\sim \omega} \left( \|p_{ij}- \mathbb{E}_{x_C} p_{i|x_C}p_{j |x_C}\|_1\right) \le \sqrt{\frac{2\log(d)}{k\Delta}}\;.
		\end{equation}
		\label{lem:yz}
	\end{lemma}
	We include a proof for completeness below, following Ref. \cite{yoshida2014approximation}.
	\begin{proof}
		For ease of notation we will use the notation $\mathbb{E}_R$ to denote expectations over random variables $R$ (i.e., our notation suppresses the distribution from which $R$ is drawn).
		
		To begin, note that for a fixed $m, C,i,j$ we have
		\[
		\|p_{ij}- \mathbb{E}_{x_C} p_{i|x_C}p_{j |x_C}\|_1=\|\mathbb{E}_{x_C}\left(p_{{ij}|{x_C}}-  p_{i|x_C}p_{j |x_C}\right)\|\leq \mathbb{E}_{x_C} \|p_{{ij}|{x_C}}- p_{i|x_C}p_{j |x_C}\|_1
		\]
		and therefore  the left hand side of Eq.~\eqref{eq:1norm} is upper bounded as
		\begin{align}
			\mathbb{E}_{m,C,i,j}  \|p_{ij}- \mathbb{E}_{x_C} p_{i|x_C}p_{j |x_C}\|_1&\leq \mathbb{E}_{m,C,i,j, x_C} \|p_{{ij}|{x_C}}-  p_{i|x_C}p_{j |x_C}\|_1\\
			&\leq \left(\mathbb{E}_{m,C,i,j, x_C}  \|p_{{ij}|{x_C}}- p_{i|x_C}p_{j |x_C}\|^2_1\right)^{1/2}
		\end{align}
		Using Pinsker's inequality we have $\|q-q'\|^2_1\leq 2D_{KL}(q||q')$ for any probability distributions $q,q'$, where $D_{KL}$ is the Kullback-Leibler divergence. Using this inequality in the above we get
		\begin{align}
			\left(\mathbb{E}_{m,C,i,j} \|p_{ij}- \mathbb{E}_{x_C} p_{i|x_C}p_{j |x_C}\|_1\right)^2&\leq \mathbb{E}_{m,C,i,j, x_C} 2D_{KL} \left(p_{{ij}|{x_C}}||p_{i|x_C}p_{j |x_C}\right)\nonumber\\
			&=2\mathbb{E}_{m,C,i,j} \left(H(i|C)-H(i|j\cup C)\right),
		\end{align}
		where $H(A|B)$ is the conditional Shannon entropy (with respect to the pseudodistribution $p$). Note that since $|C|<k<t$ and $p$ is a $t$-local pseudodistribution,  all the distributions such as $p_{ij|_{X_C}}$ discussed above are well defined.
		
		To complete the proof, below we show that 
		\begin{equation}
			\mathbb{E}_{m,C,i,j} \left(H(i|C)-H(i|j\cup C)\right)\leq \frac{\log(d)}{\Delta k}.
			\label{eq:hcond}
		\end{equation}
		To this end, write
		\begin{align}
			\mathbb{E}_{m,C,i,j} \left(H(i|C)-H(i|j\cup C)\right)&=\frac{1}{k}\sum_{0\leq m<k} \frac{1}{{n\choose m}} \sum_{|C|=m} \sum_{i,j} \omega(i,j)  \left(H(i|C)-H(i|j\cup C)\right)\\
			&\leq \frac{1}{k\Delta n^2}\sum_{0\leq m<k} \frac{1}{{n\choose m}} \sum_{|C|=m} \sum_{i,j=1}^{n}\left(H(i|C)-H(i|j\cup C)\right)
			\label{eq:sumij}
		\end{align}
		where we used the fact that $n^2\Delta \omega(i,j)\leq 1$. Next observe that terms in Eq.~\eqref{eq:sumij} with $j\in C$ are equal to zero, and therefore 
		\begin{align}
			\mathbb{E}_{m,C,i,j}& \left(H(i|C)-H(i|j\cup C)\right)\\
			&\leq \frac{1}{k\Delta n^2}\sum_{0\leq m<k} \frac{1}{{n\choose m}} \sum_{|C|=m} \sum_{i=1}^{n}\sum_{j\notin C} \left(H(i|C)-H(i|j\cup C)\right)\\
			&=\frac{1}{k\Delta n^2}\sum_{i=1}^{n} \sum_{0\leq m<k} \bigg(\sum_{|C|=m} \frac{(n-m)}{{n \choose m}} H(i|C)-\sum_{|C|=m+1} \frac{(n-m)}{{n \choose m+1}} H(i|C)\bigg).\\
			&\leq \frac{1}{k\Delta n^2}\sum_{i=1}^{n} \sum_{0\leq m<k} \bigg(\sum_{|C|=m} \frac{(n-m)}{{n \choose m}} H(i|C)-\sum_{|C|=m+1} \frac{(n-(m+1))}{{n \choose m+1}} H(i|C)\bigg).
			\label{eq:sumij2}
		\end{align}
		Now we see that the sum over $m$ is telescopic in Eq.~\eqref{eq:sumij2} and is upper bounded by 
		\[
		\frac{1}{k\Delta n^2}\sum_{i=1}^{n} \frac{(n-0)}{{n \choose 0}} H(i)\leq \frac{\log(d)}{k\Delta}
		\]
		where we upper bounded $H(i)\leq \log(d)$. We have shown Eq.~\eqref{eq:hcond}, completing the proof.
	\end{proof}
	
	\begin{lemma}\label{lem:energy_bound_dense}
		Let $H$ be a $\Delta$-dense, $2$-local, $n$-qubit Hamiltonian. Let $\sigma$ be a $k$-local pseudodensity matrix and let $\rho\in \mathcal{B}$ be its image under the rounding map Eq.~\eqref{eq:roundrho}. Suppose $\epsilon\in (0,1)$ is such that $n\geq \epsilon^{-3}\Delta^{-2}$.  We can choose $k=O(1/(\Delta \epsilon^2))$ such that
		\begin{align}
			|\trace{(\rho H)} - \sum_{i,j}\trace{(H_{ij}\sigma_{ij})} | \le \epsilon\Gamma \;.
		\end{align}
	\end{lemma}
	\begin{proof}
		Let $q$ be the $k$-local pseudodistribution obtained from $\sigma$ via the map $q_S= \Lambda^{\otimes |S|} (\sigma_S)$ for all subsets $S\subseteq [n]$ of size at most $k$. Here $\Lambda$ is the map which corresponds to choosing a single qubit Pauli basis $X,Y,Z$ uniformly at random and then measuring in that basis. The outcomes of such a single-qubit measurement are indexed by $b,r\in \{1,2,3\}\times \{0,1\}$ and so $q$ is a $k$-local pseudodistribution over $n$, $d$-ary variables with $d=6$. In particular, we may write
		\[
		q_S(x)=\frac{1}{3^{|S|}} \langle \psi_x|\sigma_S|\psi_x\rangle \qquad x\in \mathcal{I}(|S|).
		\]
		
		By definition we have 
		\begin{equation}
			q_{ij}= \Lambda^{\otimes 2}(\sigma_{ij}).
			\label{eq:lambda1}
		\end{equation}
		We can also see from Eqs.~(\ref{eq:pc},\ref{eq:rhotil}) that for any subset $C\subseteq [n]$ of size $|C|=m<k$, $x_C\in \mathcal{I}(|C|)$, and for any $i,j\notin C$ with $i\neq j$ we have
		\begin{equation}
			\Lambda (\rho^{(x_C)}_i)=q_i|_{x_C} \quad \text{and} \quad \Lambda\otimes \Lambda(\eta_{ij})=\mathbb{E}_{x_C}\left[q_i|_{x_C}q_j|_{x_C}\right].
			\label{eq:lambda2}
		\end{equation}

		Now let us use Lemma \ref{lem:yz} with $\omega(i,j)\equiv \mathrm{\Gamma_{ij}}/\Gamma$. Using the lemma and Eqs.~(\ref{eq:lambda1},\ref{eq:lambda2}) gives
		\begin{equation}
			\frac{1}{k}\sum_{0\leq m< k} \; \Ex_{\substack{C\subseteq [n]\\|C|=m}}\; \Ex_{(i,j)\sim \omega} \delta(i,j\notin C)\|\Lambda\otimes \Lambda \left(\sigma_{ij}- \eta_{ij}\right)\|_1 \le \sqrt{\frac{2\log(6)}{k\Delta}}\;,
			\label{eq:lambda3}
		\end{equation}
		where we also used the fact that $\omega(i,i)=0$. 
		
		The expression Eq.~\eqref{eq:lambda3} does not contain terms with $i\in C$ or $j\in C$, but it will be helpful to add them back in (and bound their contribution to the sum). Note that for any $i,j$ with $i\neq j$ we have $\|\Lambda\otimes \Lambda \left(\sigma_{ij}- \eta_{ij}\right)\|_1\leq 2$ by the triangle inequality. Therefore 
		\begin{align}
			\mathbb{E}_{m<k}\; \Ex_{\substack{C\subseteq [n]\\|C|=m}}\; \Ex_{(i,j)\sim \omega}& \left(1-\delta(i,j\notin C)\right)\|\Lambda\otimes \Lambda \left(\sigma_{ij}- \eta_{ij}\right)\|_1 \\
			&\le \mathbb{E}_{m<k} \; \Ex_{\substack{C\subseteq [n]\\|C|=m}}\; \sum_{i,j} \omega(i,j) 2\left(1-\delta(i,j\notin C)\right)\\
			&\leq\mathbb{E}_{m<k} \; \Ex_{\substack{C\subseteq [n]\\|C|=m}}\; \sum_{i,j} \frac{1}{\Delta n^2} 2\left(1-\delta(i,j\notin C)\right)\\
			&= \mathbb{E}_{m<k} \;  \frac{1}{\Delta n^2} 2m(2n-m)\\
			&\leq \frac{4k}{\Delta n}.
			\label{eq:fourk}
		\end{align}
		Combining Eqs.~(\ref{eq:lambda3}, \ref{eq:fourk}) and using the fact that $n\geq \epsilon^{-3}\Delta^{-2}$ gives
		\begin{equation}
			\frac{1}{k}\sum_{0\leq m< k} \; \Ex_{\substack{C\subseteq [n]\\|C|=m}}\; \Ex_{(i,j)\sim \omega} \|\Lambda\otimes \Lambda \left(\sigma_{ij}- \eta_{ij}\right)\|_1 \le \sqrt{\frac{2\log(6)}{k\Delta}}+4k\epsilon^{3}\Delta.
			\label{eq:lambda4}
		\end{equation}
		
		Following Ref. \cite{brandao2013product} we use the fact that the map $\Lambda$ does not ``distort" the $1$-norm too much:
		\begin{claim}
			For any $\ell\geq 1$ and any Hermitian $\ell$-qubit operator $Q$ we have $\|\Lambda^{\otimes \ell}(Q)\|_1\geq 6^{-\ell}\|Q\|_1$.
			\label{claim:dist}
		\end{claim}
		We give a proof of the claim below. Using the claim with $\ell=2$ and Eq.~\eqref{eq:lambda4} gives
		\begin{equation}
			\frac{1}{k}\sum_{0\leq m< k} \; \Ex_{\substack{C\subseteq [n]\\|C|=m}}\; \Ex_{(i,j)\sim \omega} \|\sigma_{ij}- \eta_{ij}\|_1 \le 36\sqrt{\frac{2\log(6)}{k\Delta}}+144k\epsilon^{3}\Delta\;.
			\label{eq:lambda5}
		\end{equation}
		Using the fact that for any Hermitian matrices $A,B$ with $\|B\|=1$ we have $\|A\|_1\geq\mathrm{Tr}(BA)$ gives 
		\begin{align}
			\frac{1}{k}\sum_{0\leq m< k} \; \Ex_{\substack{C\subseteq [n]\\|C|=m}}\; \Ex_{(i,j)\sim \omega} \|\sigma_{ij}- \eta_{ij}\|_1&\geq \left|\frac{1}{k}\sum_{0\leq m< k} \; \Ex_{\substack{C\subseteq [n]\\|C|=m}}\; \Ex_{(i,j)\sim \omega} \mathrm{Tr}(H_{ij}/\|H_{ij}\| (\sigma_{ij}-\eta_{ij}))\right|\\
			&=\frac{1}{\Gamma}\left|\mathrm{Tr}(H\rho)-\sum_{i,j} \mathrm{Tr}(H_{ij} \sigma_{ij})\right|.
			\label{eq:lambda6}
		\end{align}
		Combining Eqs.~(\ref{eq:lambda5}, \ref{eq:lambda6}) we see that we may choose a $k= O(1/(\epsilon^2\Delta))$ such that 
		\begin{align}
			|\trace(H\rho) - \sum_{i,j} \mathrm{Tr}(H_{ij} \sigma_{ij}) | \le \epsilon\Gamma.
		\end{align}
	\end{proof}
	
	\begin{proof}[Proof of Claim \ref{claim:dist}]
		The map $\Lambda$ is defined by 
		\[
		\Lambda(A)=\frac{1}{3}\sum_{b\in [3]}\sum_{r\in \{0,1\}} |b,r\rangle\langle b,r| \langle \psi_{b,r}|A|\psi_{b,r}\rangle.
		\]
		and therefore $\Lambda(P_0)=I$ and $\Lambda(P_b)=3^{-1}(|b\rangle \langle b|\otimes Z)$ for $b=1,2,3$, where $P_0=I, P_1=X, P_2=Y, P_3=Z$. Any $\ell$-qubit Hermitian operator $Q$ can be expanded in the Pauli basis as 
		\begin{equation}
			Q=\sum_{b\in \{1,2,3\}^\ell} q_{b} P(b) \qquad \qquad P(b)=P_{b_1}\otimes P_{b_2}\otimes\ldots P_{b_\ell}
		\end{equation}
		where $q_b\in \mathbb{R}$ for each $b\in [3]^\ell$. Then
		\begin{equation}
			\|Q\|_1\leq 2^{\ell/2}\|Q\|_2=2^{\ell/2}\sqrt{\mathrm{Tr}(Q^\dagger Q)}=2^{\ell}\left(\sum_{b} (q_b)^2\right)^{1/2},
			\label{eq:qex}
		\end{equation}
		and
		\begin{equation}
			\|\Lambda^{\otimes \ell}(Q)\|_1\geq \|\Lambda^{\otimes \ell} (Q)\|_2=2^{\ell/2}\left(\sum_{b} (q_b 3^{-\mathrm{wt}(b)})^2\right)^{1/2}\geq 2^{-\ell/2}3^{-\ell} \|Q\|_2
			\label{eq:lqex}
		\end{equation}
		where $\mathrm{wt}(b)$ is the number of nonzeros in $b$. Combining Eqs.~(\ref{eq:qex}, \ref{eq:lqex}) we get
		\[
		\frac{\|\Lambda^{\otimes \ell}(Q)\|_1}{\|Q\|_1}\geq 3^{-\ell}2^{-\ell}=6^{-\ell}.
		\]
	\end{proof}
	\subsection{Optimization over pseudodensity matrices\label{sec:optimize}}
	Here we review efficient algorithms for minimizing convex functions over the set of pseudodensity matrices, following Ref.~\cite{liu2006consistency} which proposed this application of the convex optimization algorithm \cite{bertsimas2004solving}.
	\begin{theorem}[Bertsimas and Vempala\cite{bertsimas2004solving}]
		Suppose $K\subseteq \mathbb{R}^m$ is a convex set and $g:K\rightarrow [-1,1]$ is a convex function. Suppose that, given $\gamma\in \mathbb{R}^{m}$, we can compute whether or not $\gamma\in K$ using a runtime at most $T$. Suppose also that $g(x)$ can be computed for a given $x$ using runtime at most $T$. Suppose $R,r\in \mathbb{R}$ and $y\in K$ are such that: $K$ is contained in the ball of radius $R$ centered at the origin, and $K$ contains the ball of radius $r$ centered at $y$.  Then we may compute $\min_{x\in K} g(x)$ to within an additive error $\delta$ using a runtime 
		\[
		O((m^5T+m^7)\log(R/r)\log(\delta^{-1})).
		\]
		\label{thm:bv}
	\end{theorem}
	
	In Section \ref{sec:dense} we use the following corollary with objective function $g\propto f_k$. Note that here and in Theorem \ref{thm:bv} it is assumed that we can decide membership in $\mathcal{B}_k$ and compute the function $g$ exactly. In practice of course we can only hope to compute using real numbers represented with $\mathrm{poly}(n)$ bits of precision. It is known that convex optimization of convex functions can be performed in the presence of such errors (see for example Theorem 4.3.13 of Ref. \cite{grotschel2012geometric}). 
	
	\begin{corol}
		Suppose $g:\mathcal{B}_k\rightarrow [-1,1]$ is a convex function which can be computed using runtime at most $T$. Then we can compute
		\[
		g^{\star}=\min_{\rho\in \mathcal{B}_k} g(\rho)
		\]
		to within an additive error $\delta$ using a runtime $n^{O(k)}T\log(\delta^{-1})$.
	\end{corol}
	\begin{proof}

		From Eqs. (\ref{eq:ps1},\ref{eq:ps2}, \ref{eq:ps3}) one can directly verify that the set of $k$-local pseudodensity matrices is convex. Let $\mathcal{P}_n$ denote the set of $n$-qubit Pauli operators $\mathcal{P}_n=\{P_1\otimes P_2\otimes \ldots P_n: \; P_i\in \{X,Y,Z,I\}\}$ and let $\mathcal{Q}_n=\mathcal{P}_n\setminus \{I\otimes I\otimes \ldots I\}$. Write $\mathrm{wt}(P)$ for the weight of the Pauli operator (number of non-identity tensor factors). We can express an $n$-qubit density matrix $\rho$ as
		\begin{equation}
			\rho=2^{-n} I+\sum_{P\in \mathcal{Q}_{n}} \alpha_P P \qquad \quad \alpha_P=2^{-n}\mathrm{Tr}(\rho_S P)
			\label{eq:pauliexpand}
		\end{equation}
		and identify $\rho$ with the real vector $\alpha$ of length $4^n-1$. Any $k$-local pseudodensity matrix $\sigma$ is identified with a real vector
		\[
		\beta_P=2^{-n} \mathrm{Tr}(\sigma P)  \qquad \quad P\in \mathcal{Q}_n, \mathrm{wt}(P)\leq k
		\]
		of length $m=\sum_{j=1}^{k} 3^j {n \choose j}=O(3^kn^k)$. In this way we can view the set of pseudodensity matrices as a convex subset $\mathcal{B}_k\subseteq \mathbb{R}^m$.
		
		If we are given a vector $\gamma\in \mathbb{R}^{m}$, we can decide whether or not $\gamma\in \mathcal{B}_k$ by checking that each reduced density matrix (on every subset $S\subseteq [n]$ of size $|S|\leq k$) defined by $\gamma$ is a valid density matrix. This is equivalent to checking that all such reduced density matrices are positive semidefinite. A $t$-qubit density matrix defined by $\gamma$ can be constructed from $O(4^t)$ entries of $\gamma$ and its eigenvalues can be computed with runtime $O(4^{3t})$. Therefore to confirm $\gamma\in \mathcal{B}_k$ takes a runtime $O({n \choose k} 4^{3k})=n^{O(k)}$.

		Therefore we can use Theorem \ref{thm:bv} to minimize the convex function $g$ over the set $B_k$. To understand its runtime we need to bound the quantities $R,r$ from Theorem \ref{thm:bv}. Note from Eq.~\eqref{eq:pauliexpand} that $\alpha_P\leq 2^{-n}$ for all $P$ and therefore a density matrix $\rho$ is identified with a vector $\alpha$ of norm at most $\|\alpha\|_2\leq 1$. Therefore any $k$-local pseudodensity matrix is also identified with a real vector $\beta$ of norm at most $R=1$. 
		
		Moreover, from Eq.~\eqref{eq:pauliexpand} and the triangle inequality we see that any vector $\alpha$ such that $\|\alpha\|_1\leq2^{-t}$ defines a valid $t$-qubit density matrix. So to ensure all $t$-qubit density matrices defined by $\beta$ are valid, for all $t\leq k$ it suffices that our vector $\beta$ satisfy $\|\beta\|_2\leq \|\beta\|_1\leq 2^{-k}$. We have shown that any vector $\beta\in \mathbb{R}^{m}$ of norm at most $2^{-k}$ defines a valid $k$-local pseudodensity matrix. Therefore we may take $r=2^{-k}$ and $y$ as the all zeros vector. The stated runtime then follows from Theorem \ref{thm:bv} with the above computed values of $m,R,r$.
		
	\end{proof}

	\section{Quantum Partition Function problem}
	\label{sec:QPF} 
	
	Recall that the partition function of a local Hamiltonian $H$ at inverse temperature $\beta>0$ is 
	\begin{align}
		Z(\beta)=\trace[e^{-\beta H}].
	\end{align}
	To simplify notation, in the following it is often convenient to absorb the inverse temperature 
	and the minus sign in the Boltzmann exponent into the definition of $H$, i.e., to consider the function $\mathrm{Tr}(e^H)$. We are interested in the problem of approximating the partition function to within a given relative error.
	
	\begin{problem}[\bf Quantum Partition Function ($k$-QPF)]
		Given a  $k$-local Hamiltonian $H$ acting on $n$ qubits
		and a precision parameter $\delta=poly(1/n)$, compute an estimate $\tilde{Z}$ such that 
		\[
		(1-\delta) \mathrm{Tr} (e^{H}) \le \tilde{Z} \le (1+\delta) \mathrm{Tr}( e^{H}).
		\]
	\end{problem}
	
	Note that the problem with any constant $\delta'\in (0,1)$ is polynomial time equivalent to the QPF problem with $\delta=\mathrm{poly}(n)^{-1}$. 
	This can be seen as follows. Consider $L$ independent copies of the system which will have a Hamiltonian $G=H\otimes \mathbb{I}^{\otimes L-1}+\mathbb{I}\otimes H\otimes \mathbb{I}^{\otimes L-2}+\cdots$. Clearly, $Z_G = {\rm Tr}(e^{-\beta G}) = Z^L$. We estimate $Z_G$ with precision $\delta'\in(0,1)$ and then take the $L$-th root of the estimate $Z_G$. We get
	\begin{align}
		(1-\delta')^{1/L}Z\leq Z_G ^{1/L} \leq Z(1+\delta')^{1/L}.
	\end{align}
	We may then choose $L=O({\rm poly}(n))$ so that $Z_G ^{1/L}$ is a $\delta={\rm poly}(n)^{-1}$ estimate of $Z$.

	In Section \ref{sec:equivalence} we establish that for any $k\geq 2$, the $k$-QPF problem is equivalent to the $2$-QPF problem. We also show that $2$-QPF is polynomial-time equivalent to three other computational problems which (respectively) involve computing
	\begin{itemize}
		\item{The mean value of a Pauli observable in the Gibbs state $\rho= e^{H}/Z$ (within a given additive error).}
		\item{The number of eigenstates of $H$ within a given interval  (within a given relative error).}
		\item{The number of witnesses that lead a QMA verifier to accept (within a given relative error).}
	\end{itemize}

	\subsection{Computational problems equivalent to QPF}
	\label{sec:equivalence}
	
	In this Section we describe three computational problems and we show that each of them is (polynomial-time) equivalent to QPF. 
	
	The first problem concerns counting the number of eigenvalues of a given local Hamiltonian (with multiplicity) inside a given interval. More precisely, given an $n$-qubit, $k$-local Hamiltonian $H$, denote its eigenvalues as $e_1\leq e_2\leq \ldots \leq e_{2^n}$. For any interval $[a,b]\in \mathbb{R}$ let 
	\[
	m_{[a,b]}=|\{i: e_i\in [a,b]\}|
	\]
	be the number of eigenvalues of $H$ in the given interval. 
	\begin{problem}[\bf Quantum Density of States ($k$-QDOS)]
		
		Given a $k$-local Hamiltonian $H$ acting on $n$ qubits, precision parameters $\delta, \epsilon=poly(1/n)$, and two thresholds $a,b$ such that $a<b$. We are asked to output an estimate $m$ satisfying
		\begin{equation}
			(1-\delta) m_{[a,b]}\leq m\leq (1+\delta)m_{[a-\epsilon, b+\epsilon]}.
			\label{eq:qdos}
		\end{equation}
	\end{problem}
	
	The second problem is to estimate the mean value of a simple (Pauli) observable with respect to the quantum Gibbs state $\rho=e^{H}/\mathrm{Tr}(e^H)$, to a given additive error.
	\begin{problem}[\bf Quantum Mean Value ($k$-QMV)]
		Given a $k$-local Hamiltonian $H$ acting on $n$ qubits, a Pauli observable
		$P$ acting non-trivially on at most $k$ qubits, and a precision 
		parameter $\epsilon=poly(1/n)$, compute an estimate $\tilde{\mu}$ such that 
		\be
		\label{mu}
		\left| \tilde{\mu} - \frac{\mathrm{Tr}( P e^H)}{\mathrm{Tr}(e^H)} \right| \le \epsilon.
		\ee
	\end{problem}
	
	The third and final problem is different from the ones defined above in that it concerns quantum circuits rather than Hamiltonians. We define the problem of \emph{quantum approximate counting} (QAC) as follows. We consider a polynomial-size \textit{verifier circuit} of the following form. It takes as input an $n$-qubit state, adjoins $n_a=\mathrm{poly}(n)$ ancilla qubits in the $|0\rangle$ state, and then applies a quantum circuit $U$ of size $\mathrm{poly}(n)$ followed by measurement of a single output qubit. Such a circuit implements an $n$-qubit two-outcome POVM $\{A,I-A\}$ where the operator $A$ corresponds to measurement outcome `1' and satisfies $0\leq A\leq I$. Formally we have
	\[
	A=\left( I\otimes\langle 0^{n_a}|\right)U^{\dagger} |1\rangle\langle1|_{\mathrm{out}} U \left(I\otimes |0^{n_a}\rangle\right).
	\]
	For any $n$-qubit input state $\psi$, the probability that it is accepted by the verifier circuit is given by $\langle \psi|A|\psi\rangle$. For any $\lambda\in [0,1]$ we write $\mathcal{L}_\lambda$ for the linear subspace spanned by all eigenstates of $A$ with eigenvalues greater than or equal to $\lambda$, and $\Pi_{\lambda}$ for the projector onto this subspace.  The dimension of this subspace is denoted
	
	\begin{equation}
		N_\lambda\equiv\mathrm{dim}(\mathcal{L}_\lambda)=\mathrm{Tr}(\Pi_{\lambda}).
		\label{eq:Nlamb}
	\end{equation}
	Informally, this is the number of witnesses accepted with probability at least $\lambda$.

	\begin{problem}[\bf Quantum Approximate Counting (QAC)]
		
		Suppose we are given a verifier circuit with $n$ input qubits and size $poly(n)$, a precision parameter $\delta=poly(1/n)$, and two thresholds $a,b$ such that $0<b<a\leq 1$ and $a-b=\Omega(1/\mathrm{poly}(n))$. We are asked to output an estimate $N$ satisfying
		\begin{equation}
			(1-\delta) N_{a} \leq N\leq (1+\delta)N_b.
			\label{eq:qac}
		\end{equation}
	\end{problem}
	
	Below we establish that all of the above problems can be reduced to one another in polynomial time. 
	\begin{theorem}
		For any $k\geq 2$, $k$-QPF is polynomial time equivalent to $2$-QPF, QAC, $k$-QDOS, and $k$-QMV. 
	\end{theorem}
	In the remainder of this section we prove the Theorem. In particular, we prove five Lemmas which establish polynomial-time reductions as depicted in Fig. \ref{fig:reductions}.

	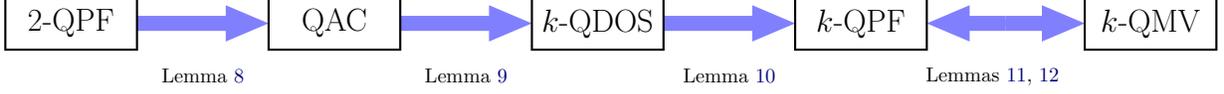
\begin{figure}
		\centering
		\begin{tikzpicture}[fill=blue!50,text opacity=1, ultra thick, scale = 0.7, transform shape,font=\Large]
			\tikzstyle{myarrows}=[line width=1mm,draw=blue!50, -triangle 45,postaction={draw, line width=2mm, shorten >=4mm, -}]
			
			\draw[myarrows]  (2.5,0.5)--(5,0.5);
			
			\draw[myarrows]  (7.5,0.5)--(10,0.5);
			
			\draw[myarrows]  (12.5,0.5)--(15,0.5);

			\draw[myarrows]  (19,0.5)--(20.5,0.5);
			\draw[myarrows]  (19,0.5)--(17.5,0.5);
			
			\draw (3.75,-0.5) node {\small Lemma \ref{lem:red1}};
			\draw (8.75,-0.5) node {\small Lemma \ref{lem:red2}};
			\draw (13.75,-0.5) node {\small Lemma \ref{lem:red3}};
			
			\draw (18.75,-0.5) node {\small Lemmas \ref{lem:red4}, \ref{lem:red5}};
			\draw[thick] (0,0) rectangle (2.5,1) node[pos=.5] {$2$-QPF};
			
			\draw[thick]   (5,0) rectangle (7.5,1) node[pos=.5] {QAC};
			\draw[thick] (10,0) rectangle (12.5,1) node[pos=.5] {$k$-QDOS};
			\draw[thick]  (15,0) rectangle (17.5,1) node[pos=.5] {$k$-QPF};
			
			\draw[thick]  (20.5,0) rectangle (23,1) node[pos=.5] {$k$-QMV};
			
		\end{tikzpicture}
		\caption{Polynomial-time reductions established in Section \ref{sec:equivalence}. An arrow $A\rightarrow B$ indicates that problem $B$ can be solved in polynomial time using an oracle which solves problem $A$.\label{fig:reductions}}
	\end{figure}
	
	\begin{lemma}\label{claim:kQPFsolvesQAC}
		QAC can be solved in polynomial time using an oracle which solves $2$-QPF.
		\label{lem:red1}
	\end{lemma}
	\begin{proof}
		Firstly note that using the standard Marriott-Watrous  in-place error reduction for the complexity class QMA \cite{marriott2005quantum} we may without loss of generality assume that $a$ is very close to $1$ and $b$ is very close to zero. That is, suppose we are given a verifier circuit with POVM $\{A,I-A\}$ and thresholds $a,b$ satisfying $a-b=\Omega(1/\mathrm{poly}(n))$. Marriott and Watrous show that for any constant $c\geq 1$ we can efficiently compute a new verifier circuit corresponding to an $n$-qubit POVM $\{A',I-A'\}$ such that
		\begin{align}
			\langle \psi|A|\psi\rangle &\geq a \rightarrow  \langle \psi|A'|\psi\rangle\geq 1-2^{-n^c}\label{eq:upperA}\\
			\langle \psi|A|\psi\rangle &\leq b \rightarrow  \langle \psi|A'|\psi\rangle\leq 2^{-n^c}.\label{eq:lowerA}
		\end{align}
		Let $a'=1-2^{-n^c}$ and $b'=2^{-n^c}$ and write $N'_{\lambda}$ for the number of witnesses accepted by $A'$ with probability at least $\lambda$ (cf.  Eq.~\eqref{eq:Nlamb}). From Eq.~\eqref{eq:upperA} we see that every state $\psi$ in the subspace $\mathcal{L}_a$, which has dimension $N_a$, satisfies $\langle \psi|A'|\psi\rangle\geq a'$. Therefore, by the Courant-Fischer min-max theorem, the operator $A'$ has at least $N_a$ eigenvalues larger than $a'$. In other words, 
		\begin{equation}
			N_a\leq N'_{a'}.
			\label{eq:naa}
		\end{equation}
		By a very similar reasoning, from Eq.~\eqref{eq:lowerA} we get 
		\begin{equation}
			N'_{b'}\leq N_b.
			\label{eq:nbb}
		\end{equation}
		Combining Eqs.~(\ref{eq:naa}, \ref{eq:nbb}) we see that we will achieve the desired estimate Eq.~\eqref{eq:qac} (for the given verification circuit with POVM $\{A,I-A\}$) if we can compute an estimate $N$ satisfying
		\begin{equation}
			(1-\delta) N'_{a'} \leq N\leq (1+\delta)N'_{b'}.
			\label{eq:qac2}
		\end{equation}
		for the new verification circuit with POVM $\{A', I-A'\}$. In the following we shall use the above with $c=1$ and assume WLOG that $a=1-2^{-n}$ and $b=2^{-n}$.
		
		We shall now make use of the $2$-local variant \cite{kempe2006complexity} of Kitaev's circuit-to-Hamiltonian mapping. This maps the given verifier circuit to a $2$-local Hamiltonian $H$ with the following properties.

		\begin{theorem}[Kempe, Kitaev, Regev \cite{kempe2006complexity}]
			\label{thm:KKR}
			Given a polynomial-size verifier circuit with $n$-qubit input, there exists an efficiently computable $2$-local Hamiltonian $H$ acting on $q(n)=\mathrm{poly}(n)$ qubits and an isometry $\mathrm{hist}: (\mathbb{C}^2)^{\otimes n}\rightarrow(\mathbb{C}^2)^{\otimes q(n)}$ such that $\langle\mathrm{hist}(\psi)|H|\mathrm{hist}(\psi)\rangle\geq 0$ for all $\psi\in (\mathbb{C}^2)^{\otimes n} $, and
			\begin{enumerate}[(i)]
				\item{Suppose $\psi$ is an $n$-qubit state that is accepted with probability at least $a$ by the verifier. Then 
					\begin{equation}
						\langle\mathrm{hist}(\psi)|H|\mathrm{hist}(\psi)\rangle\leq 1-a.
						\label{eq:enhist}
					\end{equation}
				}
				\item{Let $\mathcal{W}\subseteq (\mathbb{C}^2)^{\otimes n}$ be a subspace such that all states in $\mathcal{W}$ are accepted with probability at most $b$ by the verifier. Let $\mathcal{R}=\mathrm{hist}(\mathcal {W}^{\perp})$. Then
					\[
					\min_{\phi \in \mathcal{R}^{\perp}} \langle \phi|H|\phi\rangle\geq 1/2-b.
					\]
				}
			\end{enumerate}
			\label{thm:kkr}
		\end{theorem}
		In Appendix \ref{sec:kkr} we describe how Theorem \ref{thm:kkr} is a simple consequence of Lemma 3 from Ref. \cite{kempe2006complexity} and we also provide a description of the isometry $\mathrm{hist}$. Informally, this is the mapping from a given $n$-qubit input state $\psi$ to the ``history state" which encodes the history of the quantum computation in which the verifier circuit is applied gate-by-gate to $\psi$. 
		
		Using part (i) of Theorem \ref{thm:kkr} with $a=1-2^{-n}$ we see that there is a subspace $\mathcal{S}=\mathrm{hist}(\mathcal{L}_a)$ of dimension $N_a$ such that the restriction of $H$ to this subspace has norm upper bounded as $\|H|_{\mathcal{S}}\|\leq 2^{-n}$. Therefore
		\begin{equation}
			\mathrm{Tr}[e^{-\beta H}]\geq N_{a}e^{-\beta/2^{n}} \qquad \text{for all} \quad \beta\geq 0.
			\label{eq:nm0}
		\end{equation}
		
		Next we use part (ii) of Theorem \ref{thm:kkr} with the subspace $\mathcal{W}=\mathcal{L}_b^{\perp}$, so that $\mathcal{R}=\mathrm{hist}(\mathcal{L}_b)$. The theorem implies that the smallest eigenvalue of the restriction $H|_{\mathcal{R}^{\perp}}$ is at least $1/2-b$.  Let $B_{\mathcal{R}}$  be an orthonormal basis of $\mathcal{R}$ and let $B_{\mathcal{R}^{\perp}}$ be an orthonormal basis of $\mathcal{R}^{\perp}$. For any $\beta\geq 0$ we see that
		\begin{equation}
			\mathrm{Tr}[e^{-\beta H}]\leq \sum_{\psi \in B_{\mathcal{R}}} \langle \psi|e^{-\beta H}|\psi\rangle+\sum_{\psi \in B_{\mathcal{R^{\perp}}}} \langle \psi|e^{-\beta H}|\psi\rangle\leq N_{b}+2^{q(n)} e^{-\beta(1/2-b)}\leq N_{b}+2^{q(n)} e^{-\beta/4}
			\label{eq:nm}
		\end{equation}
		where we used part (ii) of Theorem \ref{thm:kkr} for the second term. For the first term we used the fact that $\mathrm{dim}(\mathcal{R})=N_b$, and that the restriction $H|_\mathcal{R}$ is positive semidefinite which follows from  $\langle\mathrm{hist}(\psi)|H|\mathrm{hist}(\psi)\rangle\geq 0$ for all $\psi$.
		
		Let us now fix $\beta=\mathrm{poly}(n)$ so that
		\begin{equation}
			2^{q(n)} e^{-\beta/4}\leq \delta/4 \qquad \text{and} \qquad  e^{-\beta/2^{n}}\geq 1-\delta/4.
			\label{eq:betachoice}
		\end{equation}
		Using our QPF oracle we can compute an estimate $Z$ such that
		\begin{equation}
			(1-\delta/4)\mathrm{Tr}[e^{-\beta H}]\leq Z\leq \mathrm{Tr}[e^{-\beta H}](1+\delta/4)
			\label{eq:nm2}
		\end{equation}
		Combining the above with Eqs.~(\ref{eq:nm0},\ref{eq:nm},\ref{eq:betachoice}) gives
		\begin{equation}
			(1-\delta/4)^2N_a\leq Z\leq (N_b+\delta/4)(1+\delta/4).
			\label{eq:nm3}
		\end{equation}
		
		Now consider two cases depending on the value of our estimate $Z$. Firstly suppose $Z\leq 1/2$. The lower bound in Eq.~\eqref{eq:nm3} implies that $N_a=0$ in this case (since $\delta\in (0,1)$ and $N_a$ is a nonnegative integer) so we obtain an exact solution to the given instance of QAC in this case by outputting $N=0$.
		
		Next suppose $Z\geq 1/2$. Then we output $N=Z$. The upper bound in Eq.~\eqref{eq:nm3} implies that $N_b\neq 0$ and thus $N_b\geq 1$. Substituting $\delta/4\leq N_b\delta/4$ in Eq.~\eqref{eq:nm3} then gives
		\begin{equation}
			(1-\delta)N_a\leq (1-\delta/4)^2 N_a \leq N\leq N_b(1+\delta/4)^2\leq N_b(1+\delta)
			\label{eq:nm4}
		\end{equation}
		We have shown that in either case our estimate $R$ satisfies Eq.~\eqref{eq:qac}.
	\end{proof}
	\begin{lemma}
		$k$-QDOS can be solved in polynomial time using an oracle for QAC.
		\label{lem:red2}
	\end{lemma}
	\begin{proof}
		Suppose we are given an instance of $k$-QDOS specified by a $k$-local Hamiltonian $H$, precision parameters $\delta, \epsilon$ and real numbers $a,b$. Without loss of generality we shall assume that $0\leq H< I$, which can always be achieved by shifting and scaling the given local Hamiltonian. For convenience we also assume $\delta\leq 1/2$.
		
		We shall first construct a polynomial-sized quantum verifier circuit $V$ which measures the energy of $H$ to within additive error $\epsilon/2$ and then flips an output bit (`accepts')  if the measured eigenvalue is in $[a-\epsilon/2, b+\epsilon/2]$. The verifier circuit may fail with probability at most $f=\Theta(2^{-n}/\mathrm{poly}(n))$ (we fix it below). That is, $f$ is an upper bound on the probability that, either (A) given as input an eigenstate $\psi$ of $H$ with energy $e\in [a,b]$, the verifier rejects, or (B) given as input an eigenstate $\psi'$ of $H$ with energy $e\notin [a-\epsilon,b+\epsilon]$, the verifier accepts.
		
		The verifier circuit $V$ is constructed using quantum phase estimation applied to the unitary operator $U=e^{i H}$, as described in Ref. \cite{nielsenchuang} with target precision $\epsilon/2$ and failure probability at most $1/3$. In order to reduce the failure probability to the small value $f$ we repeat the phase estimation circuit $O(n)$ times sequentially, computing the phase into $O(n)$ different registers, and then the median of these estimates is used as the output\footnote{This error reduction procedure reduces the failure probability exponentially in the number of registers. To see this, observe that if the circuit erred on the median of the estimates, then it must have erred on at least $1/2$ of the estimates. Also note that this is a different method for reducing the failure probability than the one given in Ref. \cite{nielsenchuang}.}.   Note that the phase estimation circuit, and thus also the circuit $V$, has size $O(\mathrm{poly}(n))$. Indeed, the total number of ancilla qubits used by the algorithm is $t=O(\log(n))$ and the algorithm uses $O(t^2)$ two-qubit gates along with one gate $CU$ which acts on $t+n$ qubits as $CU|j\rangle|\phi\rangle=|j\rangle  U^j|\phi\rangle$ for $0\leq j\leq 2^t-1$. The latter gate can be implemented using a $\mathrm{poly}(n)$ sized quantum circuit using well-known Hamiltonian simulation algorithms (see e.g. Ref \cite{lloyd1996universal, low2017optimal}).
		
		Let $A$ be the POVM element corresponding to measurement outcome `accept' for the above verifier circuit $V$. Recall that for $\lambda\in [0,1]$ we write $N_{\lambda}$ for the dimension of the subspace spanned by eigenvectors of $A$ with eigenvalues $\geq \lambda$, and $\Pi_{\lambda}$ for the projector onto this subspace, so that $N_{\lambda}=\mathrm{Tr}(\Pi_{\lambda})$. Note that for any $\lambda\in [0,1]$ we have the operator inequalities
		\begin{equation}
			\lambda \Pi_{\lambda}\leq A\leq \Pi_{\lambda}+(I-\Pi_{\lambda})\lambda.
			\label{eq:Aopineq}
		\end{equation}

		For each eigenvector $\psi$ of $H$ with eigenvalue in the interval $[a,b]$, the verifier circuit accepts with probability at least $1-f$, that is:
		\[
		1-f\leq \langle \psi|A|\psi\rangle\leq \langle \psi| \Pi_{\lambda}+(I-\Pi_{\lambda})\lambda|\psi\rangle,
		\]
		for any $\lambda\in [0,1]$, where we used the upper bound from Eq.~\eqref{eq:Aopineq}. Choosing $\lambda=\delta/2$ and upper bounding $\langle \psi|(I-\Pi_{\delta/2})|\psi\rangle \leq 1$ gives
		\[
		\langle \psi|\Pi_{\delta/2}|\psi\rangle\geq 1-f-\delta/2.
		\]
		Since the above holds for each vector $\psi$ in an orthonormal basis consisting of $m_{[a,b]}$ eigenvectors of $H$ with eigenvalues in $[a,b]$, we get
		\begin{equation}
			N_{\delta/2}=\mathrm{Tr}(\Pi_{\delta/2})\geq (1-f-\delta/2) m_{[a,b]}.
			\label{eq:lower}
		\end{equation}
		Next suppose $\psi'$ is any eigenvector of $H$ with eigenvalue $e$ satisfying $e\notin [a-\epsilon, b+\epsilon]$. Then the phase estimation errs if it accepts $\psi'$, that is
		\[
		f\geq \langle \psi'|A|\psi'\rangle\geq \lambda  \langle \psi'|\Pi_{\lambda}|\psi'\rangle
		\]
		for any $\lambda\in [0,1]$, where we used the lower bound from Eq.~\eqref{eq:Aopineq}. Choosing $\lambda=\delta/4$ and rearranging gives
		\begin{equation}
			\langle \psi'|\Pi_{\delta/4}|\psi'\rangle \leq \frac{4f}{\delta}
			\label{eq:psiprime}
		\end{equation}
		and therefore
		\begin{equation}
			N_{\delta/4}=\mathrm{Tr}(\Pi_{\delta/4})\leq m_{[a-\epsilon,b+\epsilon]}+4\delta^{-1}f\cdot 2^n. 
			\label{eq:upper}
		\end{equation}
		Here we used Eq.~\eqref{eq:psiprime} and the naive upper bound of $2^n$ on the number of eigenvectors of $H$ with eigenvalue outside the interval $[a-\epsilon, b+\epsilon]$.
		
		Now let us use our oracle for QAC to solve the instance defined by the verifier circuit $V$, $\delta'=\delta/4$, $a=\delta/2$ and $b=\delta/4$. This gives an estimate $N$ satisfying
		\[
		(1-\delta/4)N_{\delta/2}\leq N\leq (1+\delta/4) N_{\delta/4}
		\]
		Substituting Eqs.~(\ref{eq:lower}, \ref{eq:upper}) we arrive at
		\begin{equation}
			(1-\delta/4)(1-f-\delta/2) m_{[a,b]}\leq N\leq (1+\delta/4)\left(m_{[a-\epsilon,b+\epsilon]}+4\delta^{-1}f\cdot 2^n\right)
			\label{eq:nest}
		\end{equation}
		Now choose 
		\[
		f=\frac{\delta^2}{16\cdot 2^n}.
		\]
		Substituting in Eq.~\eqref{eq:nest} and using the loose lower bound $(1-f-\delta/2)(1-\delta/4)\geq 1-\delta$ gives
		\begin{equation}
			(1-\delta) m_{[a,b]}\leq N\leq (1+\delta/4)\left(m_{[a-\epsilon,b+\epsilon]}+\delta/4\right)
			\label{eq:nbound}
		\end{equation}
		Now we describe our estimate $m$ which solves the given instance of $k$-QDOS. We consider two cases. Firstly, if $N< (1-\delta)$ then Eq.~\eqref{eq:nbound}  implies $m_{[a,b]}=0$ and we output $m=0$, which satisfies Eq.~\eqref{eq:qdos} in this case. If instead $N\geq (1-\delta)$ then from the upper bound Eq.~\eqref{eq:nbound} we see that $m_{[a-\epsilon,b+\epsilon]}\geq 1$. Indeed, otherwise, if $m_{[a-\epsilon,b+\epsilon]}=0$ then from Eq.~\eqref{eq:nbound} we would conclude $1-\delta\leq N \leq (1+\delta/4)\delta/4$ which is not possible for any $\delta\in [0,1/2]$. Since $m_{[a-\epsilon,b+\epsilon]}\geq 1$ we have $\delta/4\leq (\delta/4) m_{[a-\epsilon,b+\epsilon]}$ and substituting on the right-hand-side of Eq.~\eqref{eq:nbound} gives
		\begin{equation}
			(1-\delta) m_{[a,b]}\leq N\leq (1+\delta/4)^2m_{[a-\epsilon,b+\epsilon]}\leq (1+\delta)m_{[a-\epsilon,b+\epsilon]}.
			\label{eq:nbound2}
		\end{equation}
		So in this case we output $m=N$ and the above shows that it satisfies Eq.~\eqref{eq:qdos} as desired.
	\end{proof}
	\begin{lemma}
		$k$-QPF can be solved in polynomial time using an oracle for $k$-QDOS.
		\label{lem:red3}
	\end{lemma}
	\begin{proof}
		Suppose we are given a $k$-local, $n$-qubit Hamiltonian $H$. We assume for convenience that $H$ is normalized so that $0\leq H< I$ and in the following we give an algorithm that produces an estimate of $Z(\beta)=\mathrm{Tr}(e^{-\beta H})$ for a given $\beta=O(\mathrm{\mathrm{poly}(n)})$. We will show that the estimate satisfies
		\begin{equation}
			(1-\delta)Z(\beta)\leq Z\leq (1+\delta) Z(\beta)
			\label{eq:estsat}
		\end{equation}
		for a constant value $\delta=0.8$. This can be amplified to the given value $\delta=1/\mathrm{poly}(n)$ as described at the beginning of Section \ref{sec:QPF}.

		Divide up the unit interval as $[0,1)=I_1\cup I_2\cup\ldots I_T$ where 
		\begin{align}
			I_j=[(j-1)/T, j/T).
		\end{align}
		Here $T=O(\mathrm{poly}(n))$ will be chosen later. We are going to use the oracle for $k$-QDOS to estimate the number of eigenvalues $m_{I_j}$ of $H$ within each of the intervals $I_j$. In particular, for each $j$ we solve the instance of  $k$-QDOS with $\delta'=0.01$, $a=(j-1)/T, b=j/T$ and $\epsilon=j/2T$. This gives estimates $m_j$ satisfying
		\[
		(0.99)m_{[\frac{j-1}{T},\frac{j}{T}]}\leq m_j\leq (1.01)m_{[\frac{j}{T}-\frac{3}{2T}, \frac{j}{T}+\frac{1}{2T}]}
		\]
		Noting that 
		\[
		\left[\frac{j}{T}-\frac{3}{2T}, \frac{j}{T}+\frac{1}{2T}\right]\subset I_{j-1}\cup I_{j}\cup I_{j+1} \qquad \text{and}\qquad I_j\subset \left[\frac{j-1}{T}, \frac{j}{T}\right]
		\]
		we arrive at
		\begin{equation}
			(0.99)m_{I_j}\leq m_j\leq (1.01)(m_{I_{j-1}}+m_{I_j}+m_{I_{j+1}}).
			\label{eq:mjbnd}
		\end{equation}
		Define
		\[
		Z=\sum_{j=1}^{T} m_j e^{-\beta(j-1)/T}.
		\]
		We have
		\[
		(0.99)Z(\beta) \leq \sum_{j=1}^{T} (0.99)m_{I_j} e^{-\beta (j-1)/T}\leq Z 
		\]
		where in the last inequality we used the lower bound from Eq.~\eqref{eq:mjbnd}. Using the upper bound from Eq.~\eqref{eq:mjbnd} gives 
		\begin{equation}
			Z \leq \sum_{j=1}^{T} (m_{I_{j-1}}+m_{I_j}+m_{I_{j+1}})(1.01)e^{-\beta(j-1)/T} \leq (1.01)\left(1+e^{\beta/T}+e^{2\beta/T}\right) Z(\beta)
		\end{equation}
		We choose $T=4\beta$ which ensures $(1.01)(1+e^{\beta/T}+e^{2\beta/T})\leq 4$ and therefore
		\[
		(0.99/4)Z(\beta)\leq Z/4\leq Z(\beta)
		\]
		Our estimate of the partition function is $Z/4$ and the above shows that it satisfies Eq.~\eqref{eq:estsat} with $\delta=1-0.99/4\leq 0.8$. As noted previously, this can be amplified to $\delta=\mathrm{poly}(n)^{-1}$.
	\end{proof}

	Next we show that the problem QPF of multiplicatively approximating  quantum partition functions
	of local Hamiltonians  is equivalent 
	to the problem QMV of additively approximating mean values of local observables
	on the thermal Gibbs state. The following lemmas provide a reduction from QMV to QPF and vice versa.
	\begin{lemma}
		$k$-QMV can be solved in time $poly(n)$ using an oracle which solves $k$-QPF.
		\label{lem:red4}
	\end{lemma}
	\begin{proof}
		Let 
		\[
		\mu =  \frac{\mathrm{Tr}( P e^H)}{\mathrm{Tr}(e^H)} 
		\]
		be the desired mean value. Define a function 
		\be
		\label{Z1}
		Z(\epsilon) =\mathrm{Tr}(e^{H+\epsilon P}).
		\ee
		Note that $H+\epsilon P$ is a $k$-local Hamiltonian. 
		Let $\delta=\epsilon^2/100$.
		Call the QPF oracle  to obtain estimates $\tilde{Z}(0)$ and $\tilde{Z}(\epsilon)$ satisfying
		\be
		\label{oracle}
		|\tilde{Z}(0) - Z(0)|\le \delta Z(0) \quad \mbox{and} \quad
		|\tilde{Z}(\epsilon) - Z(\epsilon)|\le \delta Z(\epsilon).
		\ee
		Define our estimate of the mean value $\mu$ as 
		\be
		\label{tilde_mu}
		\tilde{\mu}=\frac{\tilde{Z}(\epsilon)-\tilde{Z}(0)}{\epsilon \tilde{Z}(0)}.
		\ee
		We claim that $|\tilde{\mu}-\mu|\le \epsilon$.
		Indeed,
		Duhamel's formula for derivatives of the matrix exponential  gives
		\be
		\label{Z'}
		Z'(\epsilon):=\frac{dZ(\epsilon)}{d\epsilon} =\mathrm{Tr}(Pe^{H+\epsilon P})
		\ee
		and
		\be
		\label{Z''}
		Z''(\epsilon):=\frac{d^2Z(\epsilon)}{d\epsilon^2} =
		\int_{0}^1 dt_1 \int_{0}^{t_1} dt_2\,
		\mathrm{Tr}\left[ e^{(1-t_1+t_2)(H+\epsilon P)} P e^{(t_1-t_2)(H+\epsilon P)} P\right].
		\ee
		Note that $\mu=Z'(0)/Z(0)$. Furthermore, the standard calculus gives
		\be
		\label{Z''(s)}
		\left| Z'(0) - \frac{Z(\epsilon)-Z(0)}{\epsilon} \right|
		\le \frac{\epsilon} 2 \, \max_{0\le s\le \epsilon} \; |Z''(s)|.
		\ee
		To bound $|Z''(s)|$ we note that for any $n$-qubit positive semi-definite operators $A,B$ one has
		\be
		\label{BUAV}
		\max_{U,V} |\mathrm{Tr}(BU AV)| = \sum_{j=1}^{2^n} \lambda_j(A) \lambda_j(B),
		\ee
		where $\lambda_j$ denotes the $j$-th largest eigenvalue 
		and the maximization is over all $n$-qubit unitary operators $U,V$.
		The bound Eq.~(\ref{BUAV}) is stated as a problem on page~77 of Ref.~\cite{bhatia2013matrix}.
		Choose $A=e^{(t_1-t_2)(H+\epsilon P)}$ and $B=e^{(1-t_1+t_2)(H+\epsilon P)}$.
		Note that $\lambda_j(A)=e^{(t_1-t_2) \lambda_j(H+\epsilon P)}$ with a similar
		formula for $\lambda_j(B)$. 
		Combining Eqs.~(\ref{Z''},\ref{BUAV}) and recalling that $P$ is unitary one gets
		\be
		\label{|Z''(s)|}
		|Z''(\epsilon)| \le \int_{0}^1 dt_1 \int_{0}^{t_1} dt_2
		\sum_{j=1}^{2^n} e^{(1-t_1+t_2) \lambda_j(H+\epsilon P) + (t_1-t_2)\lambda_j(H+\epsilon P)}
		= \frac12 \sum_{j=1}^{2^n}  e^{\lambda_j(H+\epsilon P)} = 
		\frac12 Z(\epsilon).
		\ee
		Define an estimate
		\be
		\label{Z'estimate}
		\tilde{Z}'(0) = \frac{\tilde{Z}(\epsilon) - \tilde{Z}(0)}{\epsilon}.
		\ee
		Note that $\tilde{\mu}=\tilde{Z}'(0)/\tilde{Z}(0)$ while $\mu=Z'(0)/Z(0)$.
		By triangle inequality, 
		\be
		|Z'(0) - \tilde{Z}'(0)| \le \left| Z'(0) - \frac{Z(\epsilon)-Z(0)}{\epsilon} \right|
		+ \frac1{\epsilon} |\tilde{Z}(0)-Z(0)| + \frac1{\epsilon} |\tilde{Z}(\epsilon)-Z(\epsilon)|.
		\ee 
		Bound the first term using Eqs.~(\ref{Z''(s)},\ref{|Z''(s)|}). Bound the second and the third
		term using Eq.~(\ref{oracle}). We arrive at
		\be
		\label{eq1}
		|Z'(0) - \tilde{Z}'(0)| \le \frac{\epsilon} 4 Z(s) + \frac{\delta}{\epsilon} Z(0) + \frac{\delta}{\epsilon} Z(\epsilon)
		\ee
		for some $s\in [0,\epsilon]$. By Weyl's inequality,
		$\lambda_j(H+sP)\le \lambda_j(H) + \| sP \|= \lambda_j(H) + s$ for any $s\ge 0$.  Thus
		\be
		Z(s)=\sum_{j=1}^{2^n} e^{\lambda_j(H+sP)} \le e^{s} \sum_{j=1}^{2^n} e^{\lambda_j(H)} 
		= e^{s} Z(0)\le e^\epsilon Z(0).
		\ee
		Substituting this into Eq.~(\ref{eq1}) one gets
		\be
		|Z'(0) - \tilde{Z}'(0)| \le Z(0) \left( \frac{\epsilon e^{\epsilon}}4 + \frac{\delta}{\epsilon} + \frac{\delta e^{\epsilon}}{\epsilon} \right)  \le \frac{\epsilon}2 Z(0)
		\ee
		if $\delta=\epsilon^2/100$.
		We arrive at
		\be
		\label{diff_bound1}
		|\tilde{\mu}-\mu| \le \frac{\epsilon}2 + \tilde{Z}'(0) \left| \frac1{Z(0)} - \frac1{\tilde{Z}(0)}\right|
		\le  \frac{\epsilon}2 + \delta \tilde{\mu}.
		\ee
		This implies
		\be
		\label{diff_bound2}
		\tilde{\mu} \le \frac{\mu+\epsilon/2}{1-\delta} \le 2
		\ee
		for small enough $\epsilon$ 
		since $\mu\le 1$. Substituting Eq.~(\ref{diff_bound2}) into Eq.~(\ref{diff_bound1})
		gives
		$|\tilde{\mu}-\mu|\le \epsilon/2 + 2\delta \le \epsilon$.
	\end{proof}
	
	\begin{lemma}
		$k$-QPF can be solved in time $poly(n)$ using an oracle which solves $k$-QMV.
		\label{lem:red5}
	\end{lemma}
	\begin{proof}
		Define a function 
		\[
		Z(\beta)=\mathrm{Tr}(e^{\beta H}).
		\]
		Our goal is to approximate $Z(1)$
		within a multiplicative error $\delta=poly(1/n)$.
		Note that we can express $Z(1)$ as a telescoping product\footnote{Such multistage approach is often employed in Markov Chain Monte Carlo methods for estimating partition functions (see, for instance, \cite{vstefankovivc2009adaptive} and the references therein). In the methods it is assumed that one can sample from the Gibbs distributions at the different temperatures $\beta_i$ of cooling schedule with the help of certain rapidly mixing Markov chains. The ratios correspond to the means of certain random variables defined with respect to these Gibbs distributions.}
		\be
		Z(1) = 2^n\prod_{p=0}^{m-1} \frac{Z(\beta_{p+1})}{Z(\beta_p)}\quad \mbox{where} \quad \beta_p =  \frac{p}{m}
		\ee
		and estimate the ratios in each stage with sufficient precision.
		Here $m$ is a parameter to be chosen later. 
		We have
		\be
		\label{ZZbound}
		\frac{Z(\beta_{p+1})}{Z(\beta_p)} = 1 + \frac{\mathrm{Tr}(H e^{\beta_p H})}{m \mathrm{Tr}(e^{\beta_p H})}
		+ O\left( \frac{ \|H\|^2 }{m^2}\right).
		\ee
		Since $H$ is a linear combination of Pauli observables acting on at most $k$ qubits,
		one can call the QMV oracle $m\cdot poly(n)$ times to obtain estimates $\tilde{\mu_p}$ satisfying 
		\be
		\left| \tilde{\mu}_p -  \frac{\mathrm{Tr}(H e^{\beta_p H})}{\mathrm{Tr}(e^{\beta_p H})}\right| \le \delta/100.
		\ee
		Define our estimate of $Z(1)$ as 
		\be
		\tilde{Z}(1) = 2^n \prod_{p=0}^{m-1} \left( 1 + \frac{\tilde{\mu}_p}{m}\right).
		\ee
		Choose $m$ large enough so that the last term in Eq.~(\ref{ZZbound}) is at most $\delta/100m$.
		Then 
		\be
		\left| 1+ \frac{\tilde{\mu}_p}{m} -  \frac{Z(\beta_{p+1})}{Z(\beta_p)} \right| \le \frac{\delta}{10m}
		\cdot  \frac{Z(\beta_{p+1})}{Z(\beta_p)}.
		\ee
		A simple algebra then shows that $(1-\delta)Z(1) \le \tilde{Z}(1) \le (1+\delta)Z(1)$.
	\end{proof}
	
	\section{Improved classical algorithm for the QPF problem}
	\label{sec:hutch}
	In this section we describe a classical algorithm for approximating the partition
	function $\mathrm{Tr}(e^{-\beta H})$
	where $H$ is a $k$-local Hamiltonian on $n$ qubits
	such that $\beta \|H\|\le b$.
	The algorithm has
	runtime
	\be
	\label{intro2restated}
	O\left(
	(b + \log{(1/\delta)})\ell 2^n/\delta  + n^2 2^n/\delta
	+1/\delta^4\right).
	\ee
	It scales as $\tilde{O}(2^n)$, assuming that 
	$\|H\|$, $\beta$, and $1/\delta$ are at most polynomial in $n$.
	To the best of our knowledge, this achieves a nearly quadratic improvement
	over ``brute-force" approaches such as 
	the exact diagonalization (e.g. the Lanczos algorithm) or simulation of the imaginary time evolution
	which appear to have runtime  at least $4^n$.
	
	It will be convenient to state our result for a slightly more general problem -- estimating the trace of a positive semidefinite matrix. The main result of this section is as follows.
	\begin{theorem}[\bf Stochastic Trace Estimation]
		\label{thm:classical}
		Suppose $A$ is a positive semidefinite matrix\footnote{Unless stated otherwise, here and below we consider square matrices.} 
		of size $d$
		such that a matrix-vector product
		$A|\psi\ra$  can be computed classically  in time $t(A)$ for any vector $|\psi\ra$.
		There exists a classical randomized algorithm that takes as input an error tolerance  $\delta>0$ 
		and outputs an estimate $\xi(A)\in \RR$ satisfying
		\be
		\label{main_approx}
		(1-\delta) \mathrm{Tr}(A) \le \xi(A) \le (1+\delta) \mathrm{Tr}(A)
		\ee
		with probability at least $0.99$. The algorithm has runtime
		\be
		\label{runtime}
		O\left(
		\frac{t(A) + d \log^2{d} }\delta
		+
		\frac1{\delta^4}\right).
		\ee
	\end{theorem}
	Specializing this theorem to the matrix exponential function $A=e^{-\beta H}$ we obtain
	a classical algorithm for approximating 
	the quantum partition function $\mathrm{Tr}(e^{-\beta H})$
	with a small relative error.
	To simplify the notations, 
	below we absorb the inverse temperature 
	$\beta$ and the minus sign into the definition of $H$. 
	\begin{corol}
		\label{corol:exp}
		Suppose $H$ is a hermitian matrix of size $d$
		such that a matrix-vector product $H|\psi\ra$ can be computed classically in time
		$t(H)$ for any vector $|\psi\ra$. There exists a classical randomized algorithm that takes as input an error tolerance  $\delta>0$, a real number $b>0$ such that $\|H\|\le b$, and outputs an estimate $\xi'(H)\in \RR$ satisfying 
		\be
		\label{corol_approx}
		(1-\delta) \mathrm{Tr}(e^H) \le \xi'(H) \le (1+\delta) \mathrm{Tr}(e^H)
		\ee
		with probability at least $0.99$. The algorithm has runtime
		\be
		\label{corol_runtime}
		O\left(
		\frac{t(H)(b + \log{(1/\delta)})  + d \log^2{d} }\delta
		+
		\frac1{\delta^4}\right).
		\ee
	\end{corol}
	We defer the proof of the corollary until the end of this section. 
	Suppose $H$ is a $k$-local Hamiltonian on $n$ qubits, where $k=O(1)$. Let $\ell\le O(n^k)$ be the number of non-zero $k$-local terms in $H$.
	Computing the matrix-vector product for a single $k$-local term takes time  $O(2^n)$ since each term acts non-trivially on $O(1)$ qubits.
	Thus $t(H)\le O(\ell 2^n)$. Substituting this and $d=2^n$ into Eq.~(\ref{corol_runtime})
	gives a classical algorithm solving the QPF problem for $k$-local $n$-qubit Hamiltonians with the runtime Eq.~(\ref{intro2restated}).
	
	\begin{proof}[\bf Proof of Theorem~\ref{thm:classical}]
		Our main technical tools are stochastic trace estimation algorithm by Meyer et al~\cite{meyer2021hutch++}
		known as Hutch++ and unitary $2$-designs~\cite{dankert2009exact,cleve2015near}.

		Let us first summarize the Hutch++ algorithm. 
		This algorithm aims to estimate $\mathrm{Tr}(A)$ with a relative error $\delta$,
		where $A$ is a positive semidefinite real matrix of size $d$.
		The same algorithm can be applied to complex matrices using the identity
		$\mathrm{Tr}(A) = (1/2)\mathrm{Tr}(A+A^*)$.
		Fix an integer $m\gg 1$
		and choose $2m$ random vectors 
		\be
		|\psi_1\ra, |\psi_2\ra,\ldots,|\psi_{2m}\ra \in \RR^d
		\ee
		such that the entries of $|\psi_i\ra$ are $\{+1,-1\}$-valued random i.i.d.~variables.
		Define a linear subspace 
		\be
		\calQ =\mathrm{span}(A|\psi_1\ra,\ldots,A|\psi_m\ra) \subseteq \RR^d.
		\ee
		Let $Q$ be the orthogonal projector onto $\calQ$. Define
		\be
		\label{h(A)}
		h(A):=\mathrm{Tr}(QA) + \frac1m \sum_{i=m+1}^{2m} \la \psi_i|(I-Q) A(I-Q)|\psi_i\ra.
		\ee
		The following fact was proved in~\cite{meyer2021hutch++}.
		\begin{fact}[\bf Hutch++ algorithm]
			\label{fact:Hutch}
			For any parameters $\delta,\eta>0$ 
			one can choose
			\be
			m=O\left( \frac{\sqrt{\log{(1/\eta)}}}\delta + \log{(1/\eta)}\right)
			\ee
			such that the random variable $h(A)$ defined in Eq.~(\ref{h(A)}) obeys
			\be
			(1-\delta)  \mathrm{Tr}(A) \le h(A) \le (1+\delta) \mathrm{Tr}(A)
			\ee
			with probability at least $1-\eta$.
		\end{fact}
		We now determine the runtime of Hutch++.
		\begin{prop}
			\label{prop:Hutch}
			Suppose a matrix-vector product $A|\psi\ra$ can be computed in time $t(A)$
			for any vector $|\psi\ra$.
			Then Hutch++ has runtime
			\be
			\label{runtime1}
			O\left(
			\frac{t(A)}{\delta} + \frac{d}{\delta^2}
			\right)
			\ee
			for any constant failure probability $\eta>0$.
		\end{prop}
		\begin{proof}
			Since $\eta$ is a constant, we have $m=O(1/\delta)$.
			Generating $2m$ random vectors $|\psi_i\ra$ takes time $O(md)$.
			Computing $m$ vectors $A|\psi_1\ra,\ldots,A|\psi_m\ra$ takes time $m t(A)$.
			An orthonormal basis in the subspace $\calQ$ spanned by $A|\psi_1\ra,\ldots, A|\psi_m\ra$  can be found in time $O(dm^2)$  using the Gram-Schmidt process.
			Let this orthonormal basis be $|\phi_1\ra,\ldots,|\phi_\ell\ra\in \RR^d$ 
			for some $\ell\le m$
			such that 
			\[
			Q=\sum_{i=1}^\ell |\phi_i\ra\la \phi_i|.
			\]
			(Note that $\ell=m$ with high probability if $A$ has full rank.)
			Now the first term in Eq.~(\ref{h(A)}) becomes
			\[
			\mathrm{Tr}(QA) = \sum_{i=1}^\ell \la \phi_i|A|\phi_i\ra.
			\]
			By assumption, each vector $A|\phi_i\ra$ can be computed in time $t(A)$
			and thus computing $\mathrm{Tr}(QA)$ takes time $O(dm^2) + O(\ell d) + \ell t(A) \le m t(A)+ O(dm^2)$.
			Vectors 
			\[
			|\tilde{\psi}_i\ra:=(I-Q)|\psi_i\ra = |\psi_i\ra - \sum_{j=1}^\ell \la \phi_j|\psi_i\ra |\phi_j\ra
			\]
			with $m+1\le i\le 2m$ can be computed in time $O(dm^2)$.
			Finally, the second term in Eq.~(\ref{h(A)}) becomes
			\[
			\frac1m\sum_{i=m+1}^{2m} \la \tilde{\psi}_i|A|\tilde{\psi}_i\ra.
			\]
			We conclude that this term can also be computed in time $mt(A) + O(dm^2)$.
			It remains to note that $m=O(1/\delta)$.
			In the case of complex matrices $A$ the runtime of Hutch++ is still given by Eq.~(\ref{runtime1}) since $t(A+A^*)\le 2t(A)$.
		\end{proof}

		Applying Hutch++ directly to the matrix $A$ in Theorem~\ref{thm:classical}
		leads to an undesirable term in the
		runtime scaling as  $d/\delta^2$.
		To get rid of this term we shall apply
		Hutch++ to a certain compressed matrix
		of much smaller size. To this end we need
		the notion of a 
		unitary $2$-design~\cite{dankert2009exact,cleve2015near}.
		Let $U(d)$ be the group of $d\times d$ unitary matrices.
		A finite subset $\calD\subseteq U(d)$ is called
		a unitary $2$-design if for any matrices
		$P,Q\in \CC^{d\times d}$ one has
		\[
		\frac1{|\calD|} \sum_{U\in \calD} UP U^\dag \otimes 
		U Q U^\dag = 
		\int d\mu(U) \, UPU^\dag \otimes U QU^\dag
		\]
		where $d\mu(U)$ is the uniform (Haar)
		distribution on the unitary group $U(d)$.
		Thus a $2$-design reproduces the second order moments of the Haar distribution. 
		
		Assume wlog that $d=2^n$ for some integer $n$ (otherwise pad the matrix $A$ with zeros). 
		Let $\calD\subseteq U(2^n)$ be some fixed unitary $2$-design on $n$ qubits.
		For concreteness, below we assume that
		$\calD$ is the $n$-qubit Clifford group
		which is known to be a $2$-design~\cite{dankert2009exact}.
		Any $n$-qubit Clifford operator
		can be specified by $O(n^2)$ bits
		using the standard stabilizer formalism~\cite{aaronson2004improved}.
		Efficient algorithms for generation of random uniformly distributed Clifford operators
		are described in Refs.~\cite{koenig2014efficiently,bravyi2021hadamard}. In particular, the algorithm of~\cite{bravyi2021hadamard} runs in time $O(n^2)$
		and outputs a quantum circuit of size $O(n^2)$ that implements a random uniform
		element of the Clifford group. We are now ready to define a Clifford compression of a matrix.
		\begin{dfn}
			\label{cliff_est}
			Suppose  $A$ is a complex matrix of size $2^n$
			and $k\in [0,n]$ is an integer. 
			A $k$-qubit Clifford compression of $A$ 
			is defined as a matrix $\phi_U(A)$ of size $2^k$
			with matrix elements
			\be
			\label{compress1}
			\la x|\phi_U(A)|y\ra = \la 0^{n-k} x|UAU^\dag |0^{n-k}y\ra,
			\qquad x,y\in \{0,1\}^k,
			\ee
			where $U$ is picked uniformly at random
			from the $n$-qubit Clifford group.
		\end{dfn}
		The following lemma shows that
		the trace of any positive semidefinite matrix
		$A$ can be approximated with a small relative error
		by a properly normalized trace of the compressed
		matrix $\phi_U(A)$. 
		\begin{lemma}[\bf Clifford compression]\label{cliff_est_perform}
			Suppose $A$ is a positive semidefinite matrix of size $2^n$.
			Let $\phi_U(A)$ be a $k$-qubit Clifford compression of $A$
			for some $k\in [0,n]$.
			Then 
			\be
			\label{f(C)relative}
			(1-\delta) \mathrm{Tr}(A) \le 2^{n-k}\, \mathrm{Tr}(\phi_U(A)) 
			\le (1+\delta) \mathrm{Tr}(A)
			\ee 
			with probability at least $1-\eta$ as long as 
			\[
			2^k \ge  \frac1{\eta \delta^2}.  
			\]
		\end{lemma}
		\begin{proof}
			Let $I_k$ be the identity matrix of size $2^k$.
			Consider a projector
			\[
			P_U = U^\dag (|0^{n-k}\ra\la 0^{n-k}|\otimes I_k)U.
			\]
			We shall need the following fact stated as Claim 2 in Ref.~\cite{gosset2019compressed}.
			\begin{fact}
				Suppose $\calD\subseteq U(2^n)$ is 
				a unitary $2$-design on $n$ qubits. Then
				\be
				\label{exp1}
				\frac1{|\calD|}\sum_{U\in \calD} P_U = 2^{k-n} I
				\ee
				and
				\be
				\label{exp2}
				\frac1{|\calD|}\sum_{U\in \calD}
				P_U \otimes P_U = a I \otimes I + b\cdot \mathrm{SWAP},
				\ee
				for some coefficients $a,b$
				satisfying
				\be
				a\le  4^{k-n} \quad \mbox{and} \quad b\le 2^k 4^{-n}.
				\ee
				Here $\mathrm{SWAP}|x\ra \otimes |y\ra=
				|y\ra\otimes |x\ra$ for all $x,y \in \{0,1\}^n$.
			\end{fact}
			Define a random variable 
			\[
			\xi_U(A) = 2^{n-k} \mathrm{Tr}(\phi_U(A)) = 2^{n-k}\mathrm{Tr}(P_U A),
			\]
			where $U$ is picked uniformly at random from some $n$-qubit unitary $2$-design $\calD$ and $\phi_U(A)$ is defined by Eq.~(\ref{compress1}).
			From Eq.~(\ref{exp1}) one infers that $\xi_U(A)$ is an unbiased estimator of $\mathrm{Tr}(A)$, that is,
			\[
			\mu:=\frac1{|\calD|} \sum_{U\in \calD}
			\xi_U(A) = \mathrm{Tr}(A).
			\]
			From Eq.~(\ref{exp2}) one infers that the second moment of $\xi_U(A)$ is 
			\[
			\Phi:=
			\frac1{|\calD|} \sum_{U\in \calD}
			(\xi_U(A))^2 =4^{n-k} a (\mathrm{Tr}(A))^2 + 4^{n-k} b \mathrm{Tr}(A^2) \le (\mathrm{Tr}(A))^2 + 2^{-k} \mathrm{Tr}(A^2).
			\]
			Thus $\xi_U(A)$ has the variance 
			\[
			\Phi - \mu^2 \le 2^{-k} \mathrm{Tr}(A^2) \le 2^{-k} (\mathrm{Tr}(A))^2 = 2^{-k} \mu^2\le \eta \delta^2 \mu^2,
			\]
			where we used that $\mathrm{Tr}(A^2)\le (\mathrm{Tr}(A))^2$ since $A$ is positive semidefinite.
			The approximation guarantee Eq.~(\ref{f(C)relative})
			now follows from the Chebyshev inequality.
			Choosing $\calD$ as the $n$-qubit Clifford group concludes the proof. 
		\end{proof}
		We are now ready to prove Theorem~\ref{thm:classical}.
		Recall that our goal is to approximate
		$\mathrm{Tr}(A)$ within a relative error $\delta$
		with probability at least $0.99$.
		Choose $k$ as the smallest integer
		satisfying 
		\be
		\label{k1}
		2^k\ge \frac{800}{\delta^2}.
		\ee
		Choose the estimator $\xi(A)$ in the statement of the 
		theorem as 
		\[
		\xi(A)=2^{n-k}\, \mathrm{Tr}(\phi_A(U))
		\]
		where $\phi_U(A)$ is the $k$-qubit Clifford compression.
		Lemma~\ref{cliff_est_perform} implies that 
		$\xi(A)$ approximates $\mathrm{Tr}(A)$
		within a relative error $\delta/2$
		with probability at least $1-1/200$.
		Thus it suffices to show that $\mathrm{Tr}(\phi_A(U))$ can be
		approximated within a relative error $\delta/2$ 
		with probability at least $1-1/200$
		in time quoted in Eq.~(\ref{runtime}).
		To this end we employ Hutch++, as stated in Fact~\ref{fact:Hutch}
		and Proposition~\ref{prop:Hutch}
		with $A$ replaced by 
		\[
		\hat{A}\equiv \phi_U(A).
		\]
		We claim that the cost of computing a matrix-vector product $\hat{A}|\psi\ra$ with $|\psi\ra\in \CC^{2^k}$ 
		is 
		\be
		\label{hatAcost}
		t(\hat{A}) = t(A) + O(n^2 2^n) = t(A)+O(d \log^2{d}).
		\ee
		Indeed,
		by definition of $\hat{A}$ one has
		\be
		\hat{A} |\psi\ra = \la 0^{n-k}|U A U^\dag |0^{n-k} \otimes \psi\ra.
		\ee 
		Given $|\psi\ra$, one can compute the vector  $|0^{n-k} \otimes \psi\ra\in \CC^{2^n}$ in time $O(2^n)$
		by padding $|\psi\ra$ with zeros. 
		The action of any $n$-qubit Clifford operator on an $n$-qubit state vector 
		can be computed in time $O(n^2 2^n)$ by compiling $U$ into a quantum circuit with $O(n^2)$
		two-qubit gates~\cite{aaronson2004improved} and simulating the circuit 
		using the standard state vector simulator. By assumption, the matrix-vector multiplication by $A$
		can be computed in time $t(A)$. Thus one can compute
		the  vector $U A U^\dag |0^{n-k} \otimes \psi\ra$ in time $t(A)+ O(n^2 2^n)$.
		Finally, the state $\hat{A}|\psi\ra$ is obtained by restricting this vector onto the first $2^k$ coordinates. 
		This proves Eq.~(\ref{hatAcost}).
		
		Applying Fact~\ref{fact:Hutch} and Proposition~\ref{prop:Hutch}
		to the matrix $\hat{A}$ of
		size $2^k$ and using the cost of matrix-vector multiplication
		$t(\hat{A})$
		from Eq.~(\ref{hatAcost})
		one concludes that  $\mathrm{Tr}(\hat{A})$ can be estimated within a relative error $\delta/2$ and any constant failure probability
		in time 
		\be
		\frac{O(t(\hat{A}))}\delta + \frac{O(2^k)}{\delta^2}=
		\frac{O(t(A) + d \log^2{d})}\delta + \frac{O(2^k)}{\delta^2}
		=\frac{O(t(A) + d \log^2{d})}\delta+ O(1/\delta^4).
		\ee
		Here the last equality uses the bound $2^k=O(1/\delta^2)$,
		see Eq.~(\ref{k1}).
	\end{proof}

	\begin{proof}[\bf Proof of Corollary~\ref{corol:exp}]
		We shall apply Theorem~\ref{thm:classical} with $A=T_k(H)$,
		where
		\[
		T_k(x) = \sum_{p=0}^k \frac{x^p}{p!}
		\]
		is the truncated Taylor series for the function $e^x$ at $x=0$.
		We shall choose the truncation order $k$ large enough so that 
		\be
		\label{TaylorRelativeError}
		(1-\delta) e^x \le T_k(x) \le (1+\delta) e^x
		\quad \mbox{for all} \quad x\in [-b,b].
		\ee
		To this end we need the following simple lemma.
		\begin{lemma}
			\label{lemma:Taylor}
			Suppose $\epsilon>0$ and $b\ge 1$ are real numbers.
			Let $k$ be any integer such that 
			\be
			\label{orderTaylor}
			k\ge \frac{4b}{\log{(2)}}  + \frac{\log{(1/\epsilon)}}{\log{(2)}}.
			\ee
			Here we use the natural logarithm.
			Then 
			\be
			\max_{x\in [-b,b]} \left| e^x -T_k(x)\right| \le \epsilon.
			\ee
		\end{lemma}
		\begin{proof}
			Taylor's theorem implies that 
			\[
			\max_{x\in [-b,b]} \left| e^x - T_k(x)\right| \le  \frac{e^{b}}{(k+1)!} b^{k+1}:=R_k(b).
			\]
			Using the identity $(k+1)!=\Gamma(k+2)=\int_0^\infty y^{k+1} e^{-y} dy$ one gets
			\[
			\frac1{R_k(b)} =e^{-b} b^{-k-1} \int_0^\infty y^{k+1} e^{-y} dy
			= e^{-b} b \int_0^\infty z^{k+1} e^{-bz} dz\ge e^{-b} \int_2^3 z^{k+1} e^{-bz} dz
			\ge 2^k e^{-4b}.
			\]
			Here in the second equality  we performed a change of variables $y=bz$.
			The last inequality follows from the bounds $z^{k+1} \ge 2^k$ and $e^{-bz}\ge e^{-3b}$
			for $z\in [2,3]$. Thus $R_k(b)\le \epsilon$ if $k$ satisfies Eq.~(\ref{orderTaylor}).
		\end{proof}
		Set $\epsilon = \delta e^{-b}$. Then for all $x\in [-b,b]$ one has
		\[
		T_k(x) \le e^x + \epsilon  \le e^x + \epsilon e^b e^{x} =(1+\delta) e^x
		\] 
		and
		\[
		T_k(x) \ge e^x -\epsilon
		\ge e^x - \epsilon e^b e^x =(1-\delta)e^x.
		\]
		Applying Lemma~\ref{lemma:Taylor} with $\epsilon=\delta e^{-b}$
		one infers that Eq.~(\ref{TaylorRelativeError})
		is satisfied if the Taylor series is truncated at the order
		\be
		\label{k_relative_error}
		k = O(b + \log{(1/\delta)}).
		\ee
		Set $A=T_k(H)$. Since $T_k(x)$ has real coefficients,
		$A$ is a hermitian operator. 
		We claim  that $A$ is positive semidefinite.
		Indeed, the assumption $\|H\|\le b$ implies that any eigevalue of $H$
		lies in the interval $[-b,b]$. 
		From Eq.~(\ref{TaylorRelativeError}) one gets
		$T_k(x)\ge (1-\delta)e^x > 0$ for any $x\in [-b,b]$.
		Thus $A=T_k(H)$ is a positive operator.
		
		Let $\lambda_1,\ldots,\lambda_d\in [-b,b]$  be the  eigenvalues of
		$H$. Since $H$ and $A$ are diagonal in the same basis, we have
		\[
		\mathrm{Tr}(A) = \sum_{i=1}^d T_k(\lambda_i) 
		\le  (1+\delta) \sum_{i=1}^d e^{\lambda_i} = (1+\delta) \mathrm{Tr}(e^H).
		\]
		Likewise
		\[
		\mathrm{Tr}(A) = \sum_{i=1}^d T_k(\lambda_i) 
		\ge  (1-\delta) \sum_{i=1}^d e^{\lambda_i} = (1-\delta) \mathrm{Tr}(e^H).
		\]
		Here we used Eq.~(\ref{TaylorRelativeError}).
		Thus $\mathrm{Tr}(A)$ approximates $\mathrm{Tr}(e^H)$ with a relative error $\delta$.
		Since $A$ is positive semidefinite, we can estimate $\mathrm{Tr}(A)$ 
		with a relative error $\delta$ using Theorem~\ref{thm:classical}.
		Thus the desired estimator $\xi'(H)$ satisfying
		Eq.~(\ref{corol_approx}) can be chosed as
		$\xi'(H)=\xi(A)$, where
		$\xi(A)$ is the estimator of Theorem~\ref{thm:classical}.
		The runtime quoted in Eq.~(\ref{runtime}) depends on the cost of
		the
		matrix-vector multiplication $t(A)$. For any vector $|\psi\ra \in \CC^d$ one has
		\[
		A|\psi\ra = \sum_{p=0}^k \frac1{p!} H^p|\psi\ra.
		\]
		Thus computing $A|\psi\ra$ amounts to computing $k$ vectors
		$(H^p/p!)|\psi\ra$ with $p=1,\ldots,k$ and performing $k$ vector additions. 
		This takes time roughly $k(t(H) + d)$. Assume wlog that $t(H)\ge d$. Then 
		\[
		t(A) \le 2k \cdot t(H) = O(t(H)(b+ \log{(1/\delta)})).
		\]
		Here the second equality uses Eq.~(\ref{k_relative_error}).
		Substituting this into Eq.~(\ref{runtime}) gives the runtime quoted in Eq.~(\ref{corol_runtime}).
	\end{proof}
	\begin{rmk}
		The runtime of the trace estimation algorithms discussed above depends on
		the desired success probability $1-\eta$. We note that 
		it suffices to choose $\eta$  to be any constant smaller than $\frac{1}{2}$. One can then boost the success probability to any desired value $1-e^{-\ell}$ by invoking the algorithm
		only $O(\ell)$ times and outputting the median.
	\end{rmk}
	
	\section{Quantum algorithm for the QPF problem}
	\label{sec:quantum}
	
	The Clifford-compression method introduced in the previous section also leads to a qubit-efficient quantum algorithm for estimating the partition function. In this section, we consider a Hamiltonian given as a positive-weighted sum of $L=\mathrm{poly}(n)$ non-commuting projectors
	\begin{align}\label{eqn:H_sum_proj}
		H = \sum_{l=1}^L \gamma_l \Pi_l,~~~~\gamma_l \ge 0~~\text{and}~~\sum_{l=1}^L \gamma_l =1 \;,
	\end{align}
	where $\Pi_l$ acts non-trivially only on a constant number of qubits.
	These conditions also ensure that the Hamiltonian satisfies $0\le H\le 1$. Given a Hamiltonian as a sum of local terms, we can find such a decomposition by diagonalizing each local term, shifting by multiples of the identity and then re-scaling. The shifting and re-scaling does change the partition function, but this is generally unavoidable even if we only want to ensure that $H$ has bounded norm and is positive semidefinite. 
	
	We apply quantum algorithms for eigenvalue transformation to obtain a ``block-encoding'' of the operator $A=\exp(-\beta H/2)$. Combining the block-encoding of $A$ with Clifford compression, we can approximate the trace of $A$ as the transition probability on a Hilbert space with a small number of additional ancilla qubits.  The rank of the compressing Clifford projector determines the number of these additional qubits. The main result of this section is the following theorem.
	
	\begin{theorem}[Quantum partition function estimation]\label{thm:qualg}
		Let $H$ be an $n$-qubit Hamiltonian as in Eq.~\eqref{eqn:H_sum_proj} and let $\beta \ge 0$ be an inverse temperature.  Then, there is a quantum algorithm to estimate the partition function $\zee(\beta)=\mathrm{Tr}(e^{-\beta H})$ with relative error $\epsilon$ and high probability, which has running-time
		\begin{align} 
			O\left(
			\frac{1}{\epsilon} \cdot\sqrt{\frac{2^n}{\calZ(\beta)}} \cdot \mathrm{poly(n)} \cdot
			\big( \beta +\log(1/\epsilon) \big) 
			\right)
		\end{align}
		and uses a total of $n+O(\log n+\log(1/\epsilon))$ qubits.
	\end{theorem}
	Aside from the $n$ qubits used for the system Hamilanian $H$, our algorithm requires only $O(\log n+\log(1/\epsilon))$ ancilla qubits, improving upon the $n+O(\log(1/\epsilon))$ ancillas needed by an earlier algorithm due to Poulin and Wocjan~\cite{poulin09sampling}. Informally, we achieve this by using Clifford compression to avoid an intermediate step of the previous algorithm which involved purifying an $n$-qubit mixed state.
	The running-time we report here is also better than the one reported in Ref.~\cite{poulin09sampling}, but note that this is immediate using new techniques for block-encoding and eigenvalue transformations~\cite{chowdhury2017quantum,vanApeldoorn2020quantumsdpsolvers,gilyen2019singular}.
	
	To prove this theorem, we need to introduce quantum trace estimation and approximate block encoding of the imaginary time evolution $\exp(-\beta H/2)$.  This is done in the next two subsections.
	
	
	\subsection{Quantum trace estimation}
	We begin by defining the notion of block-encoding.
	\begin{dfn}\cite{gilyen2019singular}
		We say that a unitary $S$ is an $a$-ancilla block-encoding of an operator $A$ with sub-normalization factor $\alpha$ whenever 
		\begin{align}\label{eq:block_encoding}
			\left(\bra{0^{\otimes a}}\otimes I \right) S \left(\ket{0^{\otimes a}}\otimes I \right) = \frac{A}{\alpha}\;.
		\end{align}
	\end{dfn}
	The operator $A$ acts on the system register $\calH=(\CC^2)^{\otimes n}$. The unitary $S$ acts on the tensor product $\calA\otimes \calH$, which is composed of an ancilla register $\calA=(\CC^2)^{\otimes a}$ and the system register $\calH$.

	We point out that the condition $\|A\|\le 1$ is a necessary condition for such block encoding to exist as $A$ is a submatrix of the unitary matrix $S$. Our final algorithm uses an approximate block-encoding $\tilde A$ of $\exp(-\beta H/2)$ whose details we discuss in the next subsection.
	\begin{theorem}[Quantum trace estimation]\label{qalgo}
		Let $M$ be a positive semidefinite operator acting on $\calH=(\CC^2)^{\otimes n}$. Let $S$ be an $a$-ancilla unitary block encoding, with sub-normalization factor $\alpha=1$, of an operator $A$ acting on $\calH$ such that $M = A^\dagger A$.
		Then, there is a quantum algorithm that outputs an estimate of the normalized trace $\tau(M)=2^{-n} \mathrm{Tr}(M)$ within relative error $\epsilon_{\mathrm{TR}}$ with probability at least $\frac{3}{4}$. The algorithm invokes $S$ and $S^\dagger$ 
		\begin{align}\label{qalgo_runtime}
			O\left( \frac{1}{\epsilon_{\mathrm{TR}}} \cdot \frac{1}{\sqrt{\tau(M)}} \right)
		\end{align}
		many times and needs an additional $O(\mathrm{poly(n))}$ quantum gates per invocation of $S$. The total number of qubits used by the algorithm is $(n+a+k)$ where $k=O\big(\log(1/\epsilon_{\mathrm{TR}})\big)$.
	\end{theorem}
	Owing to the use of Clifford compression, our trace-estimation introduces only $\log(1/\epsilon_{TR})$ ancilla qubits, in addition to those needed by the block-encoding $S$. 
	
	For convenience, we repeat here some facts about Clifford compression (see Definition~\ref{cliff_est} and Lemma~\ref{cliff_est_perform}). 
	Let   $\calD$ be the $n$-qubit Clifford group.
	Consider a random uniformly distributed Clifford operator
	$U\in \calD$.
	Define a projector
	\begin{align}
		P_U = U^\dag ( |0^{n-k}\ra\la 0^{n-k} |\otimes I_k) U   
	\end{align}
	and an estimator 
	\begin{align}
		p_U(M)=2^{-k} \mathrm{Tr}(P_U M).
	\end{align} 
	Using Lemma~\ref{cliff_est_perform}  and the identity $p_U(M)=2^{-k}\mathrm{Tr}(\phi_U(M))$ one infers that 
	$p_U(M)$ approximates the normalized trace 
	\begin{align}
		\tau(M)=2^{-n}\mathrm{Tr}(M)
	\end{align} 
	with relative error $\epsilon_{\mathrm{CE}}$ and success probability $1-\delta_{\mathrm{CE}}$ provided that 
	\be
	\label{k}
	k\ge 2\log_2{(1/\epsilon_{\mathrm{CE}})} + \log_2{(1/\delta_{\mathrm{CE}})}.
	\ee
	Below we refer to the random variable $p_U(M)$ as a Clifford estimator.

	We decompose the system register $\calH$ as $\calH=\calH'\otimes\calH''$, where $\calH'=(\CC^2)^{\otimes (n-k)}$ and $\calH''=(\CC^2)^{\otimes k}$.  The Clifford projector $P_U$ and the Clifford operator $U$ both act on $\calH$, the projector $|0^{n-k}\ra\la 0^{n-k}|$ acts on the first subregister $\calH'$, and the identity $I_k$ acts on the second subregister $\calH''$.
	
	\begin{lemma}[Clifford estimator as transition probability]\label{cliff_est_prob}
		Let $M$ and $S$ be as in Theorem~\ref{qalgo}. Then, the normalized Clifford estimator $p_U(M)=2^{-k}\mathrm{Tr}(P_U M)$ can be expressed as a transition probability
		\begin{align}\label{eq:trace_transition_prob}
			p_U(M) &= \left\| \bra{0^a} \otimes I_n \otimes I_k) \, ( S \otimes I_k) \, |\Psi_U\> \right\|^2,
		\end{align}
		on the composite Hilbert space
		\begin{align}
			\calA \otimes \calH' \otimes \calH'' \otimes \calE
		\end{align}
		where $\ket{0^a} \in \calA$, $\calE=(\CC^2)^{\otimes k}$ is a new ancilla register and the state $\ket{\Psi_U}$ can be prepared with $O(\mathrm{poly}(n)))$ gates.
	\end{lemma}
	\begin{proof}
		The unitary $S$ acts on $\calA\otimes\calH'\otimes\calH''$.
		The state $|\Psi_U\>$ is defined as
		\begin{align}
			|\Psi_U\> &= (I_a \otimes U \otimes I_k) \, (|0^a\> \otimes |0^{n-k}\> \otimes |\Phi_k\>) 
			\, \in \,
			\calA \otimes \calH' \otimes \calH'' \otimes \calE,
		\end{align}
		where $|\Phi_k\>$ denotes the maximally entangled state
		\begin{align}
			|\Phi_k\> = \frac{1}{\sqrt{2^k}} \sum_{x\in\{0,1\}^k} |x\> \otimes |x\> 
			\, \in \,
			\calH'' \otimes \calE.
		\end{align}
		From the statement of Theorem~\ref{qalgo}, we see that
		\begin{align}
			(\bra{0^a} \otimes I_n \otimes I_k)& \, ( S \otimes I_k) \, |\Psi_U\>\nonumber \\ &= \left((\bra{0^a} \otimes I_n) \, S (I_a\otimes U) \otimes I_k \right) (|0^a\> \otimes |0^{n-k}\> \otimes |\Phi_k\>) \nonumber \\
			&= \left( AU\otimes I_k \right) \left( |0^{n-k}\> \otimes |\Phi_k\> \right)
		\end{align}
        
        \noindent
		since $S$ block-encodes $A$. Further
		\begin{align}
			\left\| AU |0^{n-k}\> \otimes |\Phi_k\> \right\|^2 &= \left(\<0^{n-k}| \otimes \< \Phi_k | \right) \left( U^\dagger A^\dagger AU\otimes I_k \right)| \left( |0^{n-k}\> \otimes |\Phi_k\> \right) \\ 
			&= 2^{-k}\trace(P_U M).
		\end{align}
		
		The cost of preparing the state $|\Psi_U\>$ is $O(\mathrm{poly}(n))$, which consists of both preparing the maximally entangled state $|\Phi_k\>$ and implementing the Clifford unitary $U$.
	\end{proof}
	
	We are now ready to complete the proof of Theorem~\ref{qalgo}.
	\begin{proof}[{\bf Proof of Theorem~\ref{qalgo}}]
		The task is to estimate the normalized trace $\tau(M)$ with relative error $\epsilon_{\mathrm{TR}}$.
		As explained in the discussion preceding Lemma~\ref{cliff_est_prob}, we use the normalized Clifford estimator $p_U(M)$ to approximate $\tau(M)$. Lemma~\ref{cliff_est_prob} equates $p_U(M)$ to the transition probability in Eq.~\eqref{eq:trace_transition_prob}. The algorithm uses quantum amplitude estimation to return an estimate $\hat p_U(M)$ of the transition probability in Eq.~\eqref{eq:trace_transition_prob} as the final output. 
		
		We first prove that our algorithm returns an accurate estimate. There are two independent sources of error: the Clifford compression and amplitude estimation.
		\begin{itemize}
			\item[(i)] For the normalized Clifford estimator $p_U(M)$, we choose the parameters $\epsilon_{\mathrm{CE}}=\epsilon_{\mathrm{TR}}/3$ and $\delta_{\mathrm{CE}}=1/8$ so that the estimate $p_U(M)$ for the normalized trace $\tau(M)$ satisfies 
			\begin{align}\label{f(C)relative_repeated}
				\left(1 - \epsilon_{\mathrm{CE}}\right) \cdot \tau(M) \le 
				p_U(M) \le 
				\left(1 + \epsilon_\mathrm{{CE}} \right) \cdot \tau(M)
			\end{align}
			with probability at least $1-\delta_\mathrm{{CE}}=7/8$. To achieve this, it suffices to set
			\begin{align}\label{eq:k_choice_quantum}
				k=O(\log(1/\epsilon_{\mathrm{TR}}))\;.
			\end{align}
			
			\item[(ii)] For the amplitude estimation algorithm in \cite[Theorem 3]{aaronson2020amplitude}, we choose the parameters
			$\epsilon_{\mathrm{AE}}=\epsilon_{\mathrm{TR}}/3$ and $\delta_{\mathrm{AE}}=1/8$ to get an estimate $\hat{p}_U(M)$ such that
			\begin{align}\label{amp_est_relative}
				\left(1 - \epsilon_{\mathrm{AE}} \right) \cdot p_U(M) < 
				\hat{p}_U(M) < 
				\left(1 + \epsilon_{\mathrm{AE}} \right) \cdot p_U(M)
			\end{align}
			holds with probability at least $1-\delta_{\mathrm{AE}}=\frac{7}{8}$.
		\end{itemize}
		
		The union bound implies that the probability that Clifford compression or amplitude estimation fails is at most $\delta_{\mathrm{CE}} + \delta_{\mathrm{AE}}=\frac{1}{4}$.  In the case that both succeed, we can bound the overall approximation error as follows
		\begin{align}
			\left| \hat{p}_U(M) - \tau(M) \right| 
			&\le
			\left| \hat{p}_U(M) - p_U(M) \right| + \left| p_U(M) - \tau(M) \right| \\
			&\le 
			\epsilon_{\mathrm{AE}} \cdot p_U(M) + \epsilon_{\mathrm{CE}} \cdot \tau(M) \\
			&\le
			\epsilon_{\mathrm{AE}} \cdot \left(1 + \epsilon_{\mathrm{CE}} \right) \cdot \tau(M) 
			+  
			\epsilon_{\mathrm{CE}} \cdot \tau(M) \\
			&\le 
			\epsilon_{\mathrm{TR}} \cdot \tau(M).
		\end{align}
		This shows that the estimator $\hat{p}_U(M)$ is correct within the desired relative precision with probability at least $\tfrac{3}{4}$.
		
		The time-complexity of this algorithm is essentially determined by that of amplitude estimation. From \cite[Theorem 3]{aaronson2020amplitude}, it follows that we need to use the unitary $S$, its inverse $S^\dagger$, and a quantum circuit that prepares $\ket{\Psi_U}$
		\begin{align}\label{ae_runtime}
			O\left( \frac{1}{\epsilon_{\mathrm{TR}}} \cdot \frac{1}{\sqrt{p_U(M)}} \right)
		\end{align}
		many times.
		The lower bound on $p_U(M)$ in (\ref{f(C)relative_repeated}) implies that $\frac{1}{2}\tau(M)\le p_U(M)$
		as $\epsilon_{\mathrm{CE}} < \frac{1}{2}$. Using this lower bound on $p_U(M)$ in (\ref{ae_runtime}), and the fact that $|\Psi_U\>$ can be prepared with $O(\mathrm{poly}(n))$ gates, we see that the number of invocations of $S$ and $S^\dagger$, as well as the number of additional gates, is as stated. The total number of qubits used is then $n+a+k$, where implementing $S$ requires $n+a$ qubits by definition, and $k=O(\log (1/\epsilon_{TR}))$ from Eq.~\eqref{eq:k_choice_quantum}.
	\end{proof}
	
	\subsection{Quantum partition function estimation}
	
	We now specialize the above results to the problem of estimating the partition function. 
	
	The operator $M$ is now an imaginary-time evolution operator $\exp(-\beta H)$ for a positive semi-definite Hamiltonian $H$. Note that the lower bound $2^n e^{-\beta}\le\mathrm{Tr}(M)$ holds since the maximal eigenvalue of $H$ is less or equal to $1$. The Lemma below shows that there is quantum circuit which is an approximate block-encoding of $A:=\exp(-\beta H/2)$ up to a constant factor.
	
	\begin{lemma}\label{lem:block_en_exp}
		Suppose $H$ is a Hamiltonian as in Eq.~\eqref{eqn:H_sum_proj}. For $\epsilon_{\mathrm{BE}} \ge 0$, there exists an efficiently computable quantum circuit which is a $O(\log n)$-ancilla block-encoding of an operator $\tilde A$, such that 
		\begin{align}
			\left\| \tilde A - \frac{1}{2}\exp(-\beta H/2) \right\| \le \epsilon_{\mathrm{BE}}\;.
		\end{align}
		The quantum circuit can be implemented with $O(\mathrm{poly(n)}(\sqrt{\beta}+\log(1/\epsilon_{\mathrm{BE}})))$ gates.
	\end{lemma}
	The proof requires us to introduce a couple of new ideas. At a high level, we use quantum eigenvalue transformation to implement an operation $\exp(-\beta H'^2/2)$, where $H'$ is an operator that acts as an ``effective square root'' of $H$.
	
	Following Ref.~\cite{somma2013spectral}, we define the operator $H'$ as
	\begin{align}
		H' = \sum_{l=1}^L \sqrt{\gamma_l} \; \Pi_l \otimes \left(\outer{l}{0}_{\mathrm{a_1}} + \outer{0}{l}_{\mathrm{a_1}} \right) \;,
	\end{align}
	where $\gamma_l$ and $\Pi_l$ are as in Eq.~\eqref{eqn:H_sum_proj}. This $H'$ is a Hamiltonian acting on a Hilbert space $\calH\otimes \calA_1$ where $\calA_1=\CC^{L+1}$ and behaves like a square root of $H$ in the following sense:
	\begin{align}\label{eq:sq_root_H}
		\left( I\otimes \bra{0}_{\mathrm{a_1}} \right) \left( H'^{2k} \right)\left( I\otimes \ket{0}_{\mathrm{a_1}} \right) = H^k ~~~\text{for all }k\in\mathbb{N}.
	\end{align}
	A simple Taylor expansion then shows that
	\begin{align}\label{eq:exp-gaussian}
		\left( I\otimes \bra{0}_{\mathrm{a_1}} \right) \left( e^{-\beta H'^2/2} \right)\left( I\otimes \ket{0}_{\mathrm{a_1}} \right) = e^{-\beta H/2}\;.
	\end{align}
	
	The eigenvalue transformation methods of Ref.~\cite{gilyen2019singular} require us to provide a unitary block-encoding of $H'$. Such a construction was shown in Ref.~\cite{chowdhury2020computing}, which we briefly outline below for completeness. We define two new operators
	\begin{align}
		&T = \left( I\otimes\ketbra{0}_{\mathrm{a_1}} +\sum_{l=1}^L \Pi_l\otimes \ketbra{l}_{\mathrm{a_1}} \right) \label{eq:sq_root_sel} \;,\\
		&G\ket{0}_{\mathrm{a_1}} = \sum_{l=1}^L \sqrt{\gamma_l} \ket{l}_{\mathrm{a_1}}
	\end{align}	
	where $G$ is unitary (recall that $\sum_l \gamma_l =1$). It can be verified that
	\begin{align}\label{eqn:sq_root_rewrite}
		H' =  T\left( I\otimes G \right) \left( I\otimes\ketbra{0}_{\mathrm{a_1}} \right) + \left(I\otimes \ketbra{0}_{\mathrm{a_1}} \right) \left( I \otimes G^{\dagger} \right) T^{\dagger}\;.
	\end{align}
	The projectors $T$ and $I\otimes \ketbra{0}_{\mathrm{a_1}}$ can be written as the sum of the identity and a reflection.
	Using these, we express $H'$ as a positive-weighted sum of at most $8=2^3$ unitaries, where the coefficients in the sum add up to $4$. Standard techniques~\cite[Lemma 5]{low2019hamiltonian} can be used to get a unitary block-encoding $W$ of $H'$ such that
	\begin{align}\label{eq:block_en_H'}
		\left(I\otimes I_{\mathrm{a_1}} \otimes \bra{0^{\otimes 3}}\right) W \left(I\otimes I_{\mathrm a_1 } \otimes\ket{0^{\otimes 3}}\right) = \frac{H'}{4}\;.
	\end{align}
	The unitary $W$ is thus a 3-ancilla block-encoding of $H'$ with sub-normalization factor 4. Moreoever, $W$ can be implemented with a quantum circuit of size $O(\mathrm{poly}(n))$. It follows from the results in Ref.~\cite{chowdhury2020computing} that if $0\le H \le 1$, then $H'$ has spectrum in $[-1,1]$.
	
	We now discuss how to block-encode the operator $\exp(-\beta H'^2/2)$ using a general result from Ref.~\cite{gilyen2019singular} which we restate below.
	\begin{lemma}[Theorem 56 in~\cite{gilyen2019singular}]\label{lem:gilyen}
		Suppose that $W$ is an $a$-ancilla block-encoding of a Hermitian matrix $H$, $-1 \le H \le 1$ with sub-normalization factor $\alpha$. If $\epsilon\geq 0$ and $ Q\in \mathbb{R}[x]$ is a degree-$k$ polynomial satisfying $| Q(x)|\leq \frac12$ for all $x\in[-1,1]$.
		Then, there is a quantum circuit $\tilde{W}$ which is an $(a+2)$-ancilla block-encoding of an operator $ A'$ with sub-normalization factor 1 such that
		\begin{align}
			\|A' - Q(H/\alpha) \| \le \epsilon\;. 
		\end{align}
		The quantum circuit $\tilde{W}$ consists of $k$ applications of $W$ and $W^\dagger$ gates, a single application of controlled-$W$ and $O((a+1)k)$ other one- and two-qubit gates, and its description can be computed with a classical computer in time $O\left(\mathrm{poly}(k,\log(1/\epsilon))\right)$.
	\end{lemma}
	
	We now have all the tools to prove Lemma~\ref{lem:block_en_exp}.
	\begin{proof}[\bf Proof of Lemma~\ref{lem:block_en_exp}]
		Let $S_{k'}(x)$ denote the truncated Taylor series of $\exp(-8\beta x^2)$ at order $k'$
		\begin{align}
			S_{k'}(x) = \sum_{j=0}^{k'} \frac{(-8\beta)^j x^{2j}}{j!}\;.
		\end{align}
		Define 
		\begin{align}\label{eq:Q_polynomial}
			Q(x) = \tfrac{1}{2}S_{2k+1}(x)
		\end{align} 
		for an integer $k$ to be specified shortly. 
		Applying Lemma~\ref{lemma:Taylor} with $x \rightarrow x^2$ and $b=2\sqrt{2\beta}$, we see that choosing $k = \Theta(\sqrt{\beta }+\log(1/\epsilon_{\mathrm{BE}}))$ is sufficient to get
		\begin{align}\label{eq:trunc_taylor_quant}
			0 \le \max_{x\in [-1,1]} \left( \frac{1}{2}\exp(-8\beta  x^2) - Q(x)\right) \le \frac{\epsilon_{\mathrm{BE}}}{2}\;.
		\end{align}
		The first inequality holds because $Q(x)$ is obtained by truncating the Taylor series of the exponential at an odd-integer term.
		Since $\exp(-\beta H'^2/2) = \exp\left[-8\beta (H'/4 )^2 \right]$, it follows from Eq.~\eqref{eq:trunc_taylor_quant} that
		\begin{align}
			\left\|Q(H'/4) - \frac{1}{2}\exp(-\beta H'^2/2) \right\| \le \frac{\epsilon_{\mathrm{BE}}}{2} \;.
		\end{align}
		Then, Eq.~\eqref{eq:trunc_taylor_quant} and that $\max_{x\in [-1,1]}|\exp(-8\beta  x^2)| \le 1$ imply $Q(x)\le 1/2$, i.e., $Q(x)$ satisfies the conditions of Lemma~\ref{lem:gilyen}. 
		Recall that the unitary $W$ in Eq.~\eqref{eq:block_en_H'} gives a 3-ancilla block-encoding of $H'$ with  sub-normalization factor $4$. We set $\epsilon =\epsilon_{\mathrm{BE}}/2$ in Lemma~\ref{lem:gilyen} and obtain the following: there is a quantum circuit $\tilde W$ which is a 5-ancilla block-encoding of an operator $A'$ such that
		\begin{align}
			\left\| A' -\frac{1}{2}\exp(-\beta H'^2/2) \right\| \le \epsilon_{\mathrm{BE}} \;.
		\end{align}
		Define $\tilde A = \left( \bra{0}_{\mathrm a_1}\otimes I \right) A' \left( \ket{0}_{\mathrm a_1}\otimes I \right)$ where $\ket{0}_{\mathrm a_1}$ is a state of $O(\log n)$ qubits. The above and Eq.~\eqref{eq:exp-gaussian} imply
		\begin{align}
			\left\| \tilde A - \frac{1}{2}\exp(-\beta H/2)  \right\| \le \epsilon_{\mathrm{BE}} \;.
		\end{align}
		The quantum circuit $\tilde W$ is then a $O(\log n)$-ancilla block-encoding of $\tilde A$. Using Lemma~\ref{lem:gilyen} and the fact that the unitary block-encoding $W$ of $H'$ has a quantum circuit of size $O(\mathrm{poly}(n))$, it follows that $\tilde W$ can be implemented with $O\left(\mathrm{poly}(n)(\sqrt{\beta }+\log(1/\epsilon_{\mathrm{BE}})) \right)$ gates.
	\end{proof}
	Finally, we prove the main result of this section.
	\begin{proof}[\bf Proof of Theorem~\ref{thm:qualg}]
		Set $\epsilon_{\mathrm{TR}}=\epsilon/3$. 
		First, we use Lemma~\ref{lem:block_en_exp} with  approximation error $\epsilon_{\mathrm{BE}}\le \frac{e^{-\beta}}{2} \cdot \epsilon_{\mathrm{TR}}$ to obtain a block-encoding $\tilde A$ such that $\|\tilde A  -\exp(-\beta H/2)\| \le \eps_{\mathrm{BE}}$ using a number of gates that is
		\begin{align}\label{eq:gates_block_en_exp_relative}
			O(\mathrm{poly(n)}(\beta+\log(1/\epsilon))).
		\end{align}
		Second, we apply the quantum trace estimation algorithm in Theorem~\ref{qalgo} to the positive semidefinite operator $M=\tilde{A}$ and relative error $\epsilon_{\mathrm{TR}}$.
		Let $\zeta$ denote the estimate of the normalized trace $\tau(\tilde{A})$ output by the quantum trace estimation.
		We obtain
		\begin{align}\label{deviation2}
			| \zeta - \tau(A)| 
			&\le 
			| \zeta - \tau(\tilde{A})| +
			| \tau(\tilde{A}) - \tau(A)| \\
			&\le
			\epsilon_{\mathrm{TR}} \cdot \tau(\tilde{A}) + \epsilon_{\mathrm{BE}} \\
			&\le
			\epsilon_{\mathrm{TR}} \cdot ( 1 + \epsilon_{\mathrm{TR}}) \cdot \tau(A)  + \epsilon_{\mathrm{TR}} \cdot \tau(A) \\
			&\le
			3 \cdot \epsilon_{\mathrm{TR}} \cdot \tau(A) \\
			&\le
			\epsilon \cdot \tau(A)
		\end{align}
		Our choice of $\epsilon_{\mathrm{BE}}$ ensures $\frac{1}{2}\tau(A)\le \tau(\tilde{A})$. 
		
		Therefore, the number of times the block encoding of $\tilde{A}$ has to be invoked inside quantum trace estimation is $O\Big( 1/ \big(\epsilon \sqrt{\tau(A)} \big) \Big)$.
		Combining this with the cost of realizing the block encoding of $\tilde{A}$ in Eq.~\eqref{eq:gates_block_en_exp_relative} gives the overall running-time as in the theorem statement.  The desired estimator for the partition function $\zee$ is $2^{n+1} \zeta$. 
		
		For the qubit count, note that we need $n$ qubits for the system Hamiltonian, and $O(\log n)$ ancilla qubits to implement the unitary block-encoding of $\tilde A$. Lastly, a further $O(\log(1/\epsilon))$ qubits are needed for quantum trace-estimation. 
	\end{proof}
	
	\section{Acknowledgments}
	DG thanks Bill Fefferman, Sevag Gharibian, and Robin Kothari for discussions about quantum approximate counting.
	SB, AC, and DG are supported in part by the Army Research Office under Grant Number W911NF-20-1-0014. DG is a CIFAR fellow in the quantum information science program, and is also supported in part by IBM Research. SB is supported in part by the IBM Research Frontiers Institute. 
	
	\appendix
	\section{Proof of Theorem \ref{thm:kkr}\label{sec:kkr}}

	Let $\calQ_k\equiv (\CC^2)^{\otimes k}$ be the Hilbert space describing a $k$-qubit register.
	Suppose $U=U_T \cdots U_2 U_1$ is the verifier circuit acting on
	$\calQ_n \otimes \calQ_{n_a}$.
	For each $t\in \{0,1,\ldots,T\}$ define a $T$-qubit clock state
	\be
	\label{unary_clock}
	|\hat{t}\rangle = |\underbrace{11\ldots1}_{t}\underbrace{00\ldots 0}_{T-t}\ra \in \calQ_T.
	\ee
	The desired isometry 
	$\mathrm{hist}\, : \, \calQ_n\to \calQ_n\otimes \calQ_{n_a} \otimes \calQ_T$ is defined as
	\be
	\label{isometry_hist}
	|\mathrm{hist}(\psi)\ra = \frac1{\sqrt{T+1}} \sum_{t=0}^T (U_t \cdots U_2 U_1 |\psi\ra \otimes
	|0^{n_a}\ra) \otimes  |\hat{t}\ra,
	\ee
	where $|\psi\ra \in \calQ_n$ is an arbitrary state.
	Accordingly, we set $q(n)=n+n_a+T$.
	Let 
	\[
	\calS_{hist}=\mathrm{span}(|\mathrm{hist}(\psi)\ra \, : \, |\psi\ra \in \calQ_n)
	\]
	be the image of this isometry.
	Below we write $\lambda(H)$ for the smallest eigenvalue of a Hamiltonian $H$.
	We write $H|_{\calS}$ for the restriction of a Hamiltonian $H$ onto a subspace $\calS$.
	We shall use the following facts  established
	by Kempe, Kitaev, and Regev~\cite{kempe2006complexity}.
	\begin{lemma}[\bf Projection Lemma~\cite{kempe2006complexity}]
		\label{lemma:KKRprojection}
		Suppose $G_1$ and $G_2$ are Hamiltonians acting on the same Hilbert space 
		$\calH=\calS\oplus \calS^\perp$ such that 
		$\calS$ is an eigenspace of $G_2$ with the zero eigenvalue 
		and  any eigenvector of $G_2$ in $\calS^\perp$ has the eigenvalue at least $J>2\|G_1\|$.
		Then
		\be
		\label{KKRlower}
		\lambda(G_1+G_2) \ge \lambda(G_1|_\calS) - \frac{\|G_1\|^2}{J-2\|G_1\|}.
		\ee
	\end{lemma}
	\begin{lemma}[\bf 2-local history Hamiltonian~\cite{kempe2006complexity}]
		\label{lemma:KKR2local}
		Suppose the verifier circuit $U$ consists of single-qubit and controlled-Z gates.
		For any $J\le poly(n,T)$
		there exists a $2$-local Hamiltonian $H_{hist}$
		acting on $\calQ_n\otimes \calQ_{n_a} \otimes \calQ_T$
		such that $\la \phi|H_{hist}|\phi\ra=0$ for all $|\phi\ra\in \calS_{hist}$
		and 
		\be
		\label{KKR2local}
		\lambda(H_{hist}+H_{out})\ge \lambda(H_{out}|_{\calS_{hist}})-\frac14
		\ee
		for any hermitian operator $H_{out}$ 
		with  $\|H_{out}\|\le J$.
	\end{lemma}
	\noindent
	(Although Lemma~\ref{lemma:KKR2local} is not explicitly stated
	in~\cite{kempe2006complexity}, it is a straightforward corollary
	of the proof of Lemma~3 thereof\footnote{Using the notations
		of Ref.~\cite{kempe2006complexity},
		$H_{hist} = J_{in} H_{in} + J_1 H_{prop1} + J_2 H_{prop2} + J_{clock}H_{clock}$.}
	). 
	
	We shall choose a
	2-local Hamiltonian satisfying conditions (i) and (ii) 
	of Theorem~\ref{thm:KKR} as
	\be
	H=H_{hist} + (T+1) I \otimes  |0\ra\la 0|_{out} \otimes  |1\ra\la 1|_T
	\ee
	where $H_{hist}$ is the Hamiltonian of Lemma~\ref{lemma:KKR2local}
	with the parameter $J$ to be chosen later. 
	The second term in $H$ acts trivially on the witness register $\calQ_n$,
	the projector $|0\ra\la 0|_{out}$ acts on the output qubit measured by the verifier,
	and the projector $|1\ra\la 1|_T$
	acts on the last qubit of the clock register $\calQ_T$.
	Recall that the verifier accepts a witness if the measurement of the
	output qubit has outcome `1'. 
	Thus the second term in $H$
	penalizes witness
	states rejected by the verifier.
	Since the second term in $H$ is positive semidefinite, one
	has 
	$\la \mathrm{hist}(\psi)|H|\mathrm{hist}(\psi)\ra\ge 
	\la \mathrm{hist}(\psi)|H_{hist}|\mathrm{hist}(\psi)\ra
	=0$ for all $|\psi\ra\in \calQ_n$.
	Here the last equality follows from Lemma~\ref{lemma:KKR2local}.
	
	Let us verify condition~(i) of Theorem~\ref{thm:KKR}.
	Suppose the verifier accepts a witness state $|\psi\ra\in \calQ_n$ with the probability at least $a$.
	Equivalently, $\la \psi|A|\psi\ra\ge a$ where 
	\[
	A=\left( I\otimes\langle 0^{n_a}|\right)U^{\dagger} |1\rangle\langle1|_{\mathrm{out}} U \left(I\otimes |0^{n_a}\rangle\right).
	\]
	By Lemma~\ref{lemma:KKR2local},
	\[
	\la \mathrm{hist}(\psi)|H|\mathrm{hist}(\psi)\ra
	=(T+1)  \| \la \mathrm{hist}(\psi)| I \otimes |0\ra_{out} \otimes |1\ra_T\|^2
	=\la \psi |I-A|\psi\ra \le 1-a.
	\]
	Here we noted that
	the projector $|1\ra\la 1|_T$ acting on the clock register
	annihilates all clock states $|\hat{t}\ra$ with $t<T$
	and acts trivially on the clock state $|\hat{T}\ra$, see Eqs.~(\ref{unary_clock},\ref{isometry_hist}).
	Thus $H$ obeys condition~(i).
	
	Let us verify condition~(ii) of Theorem~\ref{thm:KKR}.
	By assumption, we are given a subspace $\calW\subseteq \calQ_n$ such that
	the verifier accepts any state $|\psi\ra\in \calW$
	with the probability at most $b$.
	Equivalently, $\la \psi |A|\psi\ra \le b$ for all $|\psi\ra\in \calW$. 
	Consider a subspace 
	\[
	\calR=\mathrm{hist}(\calW^\perp)\subseteq \calQ_n \otimes \calQ_{n_a} \otimes \calQ_T
	\]
	and let $\Pi_{\calR}$ be the projector onto $\calR$.
	Choose the operator $H_{out}$ in Lemma~\ref{lemma:KKR2local} as
	\be
	\label{Hout}
	H_{out} =  (T+1) I \otimes  |0\ra\la 0|_{out} \otimes  |1\ra\la 1|_T + 6\Pi_{\calR}.
	\ee
	From Eqs.~(\ref{unary_clock},\ref{isometry_hist}) one gets
	\be
	\la \mathrm{hist}(\psi)|H_{out}|\mathrm{hist}(\phi)\ra= \la\psi|I-A +6\Pi_{\calW^\perp}|\phi\ra
	\ee
	for all states $\psi,\phi\in \calQ_n$, where $\Pi_{\calW^\perp}$
	is the projector onto $\calW^\perp$.
	Equivalently, 
	$H_{out}|_{\calS_{hist}}=I-A+6\Pi_{\calW^\perp}$. Thus
	\[
	\lambda(H_{out}|_{\calS_{hist}}) = \lambda(I-A + 6\Pi_{\calW^\perp}).
	\]
	Applying the Projection Lemma with $G_1=I-A$, $G_2=6\Pi_{\calW^\perp}$, and $\calS=\calW$ one gets
	\be
	\lambda(H_{out}|_{\calS_{hist}}) = \lambda(I-A+6\Pi_{\calW^\perp})
	\ge \lambda((I-A)|_\calW) - \frac14.
	\ee
	Here we noted that $\|G_1\|\le 1$.
	Combining this and Eq.~(\ref{KKR2local}) of Lemma~\ref{lemma:KKR2local} gives
	\[
	\lambda(H_{hist}+H_{out})\ge 
	\lambda(H_{out}|_{\calS_{hist}})-\frac14
	\ge \lambda((I-A)|_\calW)-\frac12\ge \frac12-b
	\]
	if we choose the parameter $J$ of Lemma~\ref{lemma:KKR2local}
	as $J=\|H_{out}\| \le (T+1)+6=T+7$, see Eq.~(\ref{Hout}).
	Consider any normalized state $|\phi\ra \in \calR^\perp$.
	By definition, $\la \phi|\Pi_{\calR}|\phi\ra=0$
	and thus 
	\be
	\la \phi| H |\phi\ra = \la \phi|H_{hist} + H_{out}|\phi\ra
	\ge \lambda(H_{hist} + H_{out})
	\ge \frac12-b.
	\ee
	Thus $H$ satisfies condition~(ii) of Theorem~\ref{thm:KKR}.

	\bibliographystyle{ieeetr}
	\bibliography{bibliog}

\begin{thebibliography}{10}

\bibitem{kitaev2002book}
A.~Y. Kitaev, A.~H. Shen, and M.~N. Vyalyi, {\em Classical and Quantum Computation}.
\newblock USA: American Mathematical Society, 2002.

\bibitem{kempe2006complexity}
J.~Kempe, A.~Kitaev, and O.~Regev, ``The complexity of the local hamiltonian problem,'' {\em Siam journal on computing}, vol.~35, no.~5, pp.~1070--1097, 2006.

\bibitem{brown2011computational}
B.~Brown, S.~T. Flammia, and N.~Schuch, ``Computational difficulty of computing the density of states,'' {\em Physical review letters}, vol.~107, no.~4, p.~040501, 2011.

\bibitem{dinur2007pcp}
I.~Dinur, ``The {PCP} theorem by gap amplification,'' {\em Journal of the ACM (JACM)}, vol.~54, no.~3, pp.~12--es, 2007.

\bibitem{alimonti1997hardness}
P.~Alimonti and V.~Kann, ``Hardness of approximating problems on cubic graphs,'' in {\em Italian Conference on Algorithms and Complexity}, pp.~288--298, Springer, 1997.

\bibitem{suzuki1993quantum}
M.~Suzuki, {\em Quantum {M}onte {C}arlo methods in condensed matter physics}.
\newblock World scientific, 1993.

\bibitem{bravyi2006complexity}
S.~Bravyi, D.~P. Divincenzo, R.~I. Oliveira, and B.~M. Terhal, ``The complexity of stoquastic local {H}amiltonian problems,'' {\em arXiv preprint quant-ph/0606140}, 2006.

\bibitem{crosson2021rapid}
E.~Crosson and A.~W. Harrow, ``Rapid mixing of path integral {M}onte {C}arlo for 1{D} stoquastic {H}amiltonians,'' {\em Quantum}, vol.~5, p.~395, 2021.

\bibitem{bravyi2017polynomial}
S.~Bravyi and D.~Gosset, ``Polynomial-time classical simulation of quantum ferromagnets,'' {\em Physical review letters}, vol.~119, no.~10, p.~100503, 2017.

\bibitem{barvinok2016combinatorics}
A.~Barvinok, {\em Combinatorics and complexity of partition functions}, vol.~9.
\newblock Springer, 2016.

\bibitem{harrow2020classical}
A.~W. Harrow, S.~Mehraban, and M.~Soleimanifar, ``Classical algorithms, correlation decay, and complex zeros of partition functions of quantum many-body systems,'' in {\em Proceedings of the 52nd Annual ACM SIGACT Symposium on Theory of Computing}, pp.~378--386, 2020.

\bibitem{arora1999polynomial}
S.~Arora, D.~Karger, and M.~Karpinski, ``Polynomial time approximation schemes for dense instances of {NP}-hard problems,'' {\em Journal of computer and system sciences}, vol.~58, no.~1, pp.~193--210, 1999.

\bibitem{brandao2013product}
F.~G. Brandao and A.~W. Harrow, ``Product-state approximations to quantum ground states,'' in {\em Proceedings of the Forty-Fifth Annual ACM Symposium on Theory of Computing}, STOC '13, (New York, NY, USA), p.~871–880, Association for Computing Machinery, 2013.

\bibitem{gharibian2012approximation}
S.~Gharibian and J.~Kempe, ``Approximation algorithms for {QMA}-complete problems,'' {\em SIAM Journal on Computing}, vol.~41, no.~4, pp.~1028--1050, 2012.

\bibitem{risteski2016partition}
A.~Risteski, ``How to calculate partition functions using convex programming hierarchies: provable bounds for variational methods,'' in {\em 29th Annual Conference on Learning Theory} (V.~Feldman, A.~Rakhlin, and O.~Shamir, eds.), vol.~49 of {\em Proceedings of Machine Learning Research}, (Columbia University, New York, New York, USA), pp.~1402--1416, PMLR, 23--26 Jun 2016.

\bibitem{yoshida2014approximation}
Y.~Yoshida and Y.~Zhou, ``Approximation schemes via sherali-adams hierarchy for dense constraint satisfaction problems and assignment problems,'' in {\em Proceedings of the 5th Conference on Innovations in Theoretical Computer Science}, ITCS '14, (New York, NY, USA), p.~423–438, Association for Computing Machinery, 2014.

\bibitem{bertsimas2004solving}
D.~Bertsimas and S.~Vempala, ``Solving convex programs by random walks,'' {\em Journal of the ACM (JACM)}, vol.~51, no.~4, pp.~540--556, 2004.

\bibitem{grotschel2012geometric}
M.~Gr{\"o}tschel, L.~Lov{\'a}sz, and A.~Schrijver, {\em Geometric algorithms and combinatorial optimization}, vol.~2.
\newblock Springer Science \& Business Media, 2012.

\bibitem{barak2011rounding}
B.~Barak, P.~Raghavendra, and D.~Steurer, ``Rounding semidefinite programming hierarchies via global correlation,'' {\em CoRR}, vol.~abs/1104.4680, 2011.

\bibitem{brandao08}
F.~G. S.~L. Brand\~{a}o, {\em Entanglement Theory and the Quantum Simulation of Many-Body Physics}.
\newblock PhD thesis, Imperial College of Science, Technology and Medicine, 2008.

\bibitem{cade2017quantum}
C.~Cade and A.~Montanaro, ``The quantum complexity of computing schatten $ p $-norms,'' {\em arXiv preprint arXiv:1706.09279}, 2017.

\bibitem{chowdhury2020computing}
A.~N. Chowdhury, R.~D. Somma, and Y.~Subasi, ``Computing partition functions in the one-clean-qubit model,'' {\em Phys. Rev. A}, vol.~103, p.~032422, Mar 2021.

\bibitem{knill1998power}
E.~Knill and R.~Laflamme, ``Power of one bit of quantum information,'' {\em Physical Review Letters}, vol.~81, no.~25, p.~5672, 1998.

\bibitem{stockmeyer1983complexity}
L.~Stockmeyer, ``The complexity of approximate counting,'' in {\em Proceedings of the Fifteenth Annual ACM Symposium on Theory of Computing}, STOC '83, (New York, NY, USA), p.~118–126, Association for Computing Machinery, 1983.

\bibitem{goldberg2017complexity}
L.~A. Goldberg and H.~Guo, ``The complexity of approximating complex-valued {I}sing and {T}utte partition functions,'' {\em computational complexity}, vol.~26, no.~4, pp.~765--833, 2017.

\bibitem{aharonov2008pursuit}
D.~Aharonov, M.~Ben-Or, F.~G. Brandao, and O.~Sattath, ``The pursuit for uniqueness: extending {V}aliant-{V}azirani theorem to the probabilistic and quantum settings,'' {\em arXiv preprint arXiv:0810.4840}, 2008.

\bibitem{cubitt2018universal}
T.~S. Cubitt, A.~Montanaro, and S.~Piddock, ``Universal quantum {H}amiltonians,'' {\em Proceedings of the National Academy of Sciences}, vol.~115, no.~38, pp.~9497--9502, 2018.

\bibitem{zhou2021strongly}
L.~Zhou and D.~Aharonov, ``Strongly universal hamiltonian simulators,'' {\em arXiv preprint arXiv:2102.02991}, 2021.

\bibitem{kohler2020translationally}
T.~Kohler, S.~Piddock, J.~Bausch, and T.~Cubitt, ``Translationally-invariant universal quantum hamiltonians in 1d,'' {\em arXiv preprint arXiv:2003.13753}, 2020.

\bibitem{kohler2021general}
T.~Kohler, S.~Piddock, J.~Bausch, and T.~Cubitt, ``General conditions for universality of quantum hamiltonians,'' {\em arXiv preprint arXiv:2101.12319}, 2021.

\bibitem{gosset2019compressed}
D.~Gosset and J.~Smolin, ``A compressed classical description of quantum states,'' in {\em 14th Conference on the Theory of Quantum Computation, Communication and Cryptography (TQC 2019)}, Schloss Dagstuhl-Leibniz-Zentrum fuer Informatik, 2019.

\bibitem{huang2020predicting}
H.-Y. Huang, R.~Kueng, and J.~Preskill, ``Predicting many properties of a quantum system from very few measurements,'' {\em Nature Physics}, vol.~16, no.~10, pp.~1050--1057, 2020.

\bibitem{cleve2015near}
R.~Cleve, D.~Leung, L.~Liu, and C.~Wang, ``Near-linear constructions of exact unitary 2-designs,'' {\em arXiv preprint arXiv:1501.04592}, 2015.

\bibitem{dankert2009exact}
C.~Dankert, R.~Cleve, J.~Emerson, and E.~Livine, ``Exact and approximate unitary 2-designs and their application to fidelity estimation,'' {\em Physical Review A}, vol.~80, no.~1, p.~012304, 2009.

\bibitem{hutchinson1989stochastic}
M.~F. Hutchinson, ``A stochastic estimator of the trace of the influence matrix for {L}aplacian smoothing splines,'' {\em Communications in Statistics-Simulation and Computation}, vol.~18, no.~3, pp.~1059--1076, 1989.

\bibitem{meyer2021hutch++}
R.~A. Meyer, C.~Musco, C.~Musco, and D.~P. Woodruff, ``Hutch++: Optimal stochastic trace estimation,'' in {\em Symposium on Simplicity in Algorithms (SOSA)}, pp.~142--155, SIAM, 2021.

\bibitem{poulin09sampling}
D.~Poulin and P.~Wocjan, ``Sampling from the thermal quantum gibbs state and evaluating partition functions with a quantum computer,'' {\em Phys. Rev. Lett.}, vol.~103, p.~220502, 2009.

\bibitem{brassard2002amplitude}
G.~Brassard, P.~H{\o}yer, M.~Mosca, and A.~Tapp, ``Quantum amplitude amplification and estimation,'' in {\em Quantum computation and information}, vol.~305 of {\em Contemporary Mathematics}, pp.~53--74, AMS, 2002.

\bibitem{chowdhury2017quantum}
A.~N. Chowdhury and R.~D. Somma, ``Quantum algorithms for gibbs sampling and hitting-time estimation,'' {\em Quant. Inf. Comp.}, vol.~17, no.~1–2, p.~41–64, 2017.

\bibitem{gilyen2019singular}
A.~Gily\'{e}n, Y.~Su, G.~H. Low, and N.~Wiebe, ``Quantum singular value transformation and beyond: Exponential improvements for quantum matrix arithmetics,'' in {\em ACM Symposium on Theory of Computing 2019}, (New York, NY, USA), p.~193–204, Association for Computing Machinery, 2019.

\bibitem{vanApeldoorn2020quantumsdpsolvers}
J.~van Apeldoorn, A.~Gily{\'{e}}n, S.~Gribling, and R.~de~Wolf, ``Quantum {SDP}-{S}olvers: {B}etter upper and lower bounds,'' {\em {Quantum}}, vol.~4, p.~230, Feb. 2020.

\bibitem{nielsenchuang}
M.~A. Nielsen and I.~L. Chuang, {\em Quantum Computation and Quantum Information: 10th Anniversary Edition}.
\newblock USA: Cambridge University Press, 10th~ed., 2011.

\bibitem{watrous_2018}
J.~Watrous, {\em The Theory of Quantum Information}.
\newblock Cambridge University Press, 2018.

\bibitem{liu2006consistency}
Y.-K. Liu, ``Consistency of local density matrices is qma-complete,'' in {\em Approximation, randomization, and combinatorial optimization. algorithms and techniques}, pp.~438--449, Springer, 2006.

\bibitem{marriott2005quantum}
C.~Marriott and J.~Watrous, ``Quantum arthur--merlin games,'' {\em computational complexity}, vol.~14, no.~2, pp.~122--152, 2005.

\bibitem{lloyd1996universal}
S.~Lloyd, ``Universal quantum simulators,'' {\em Science}, pp.~1073--1078, 1996.

\bibitem{low2017optimal}
G.~H. Low and I.~L. Chuang, ``Optimal hamiltonian simulation by quantum signal processing,'' {\em Physical review letters}, vol.~118, no.~1, p.~010501, 2017.

\bibitem{bhatia2013matrix}
R.~Bhatia, {\em Matrix analysis}, vol.~169.
\newblock Springer Science \& Business Media, 2013.

\bibitem{vstefankovivc2009adaptive}
D.~{\v{S}}tefankovi{\v{c}}, S.~Vempala, and E.~Vigoda, ``Adaptive simulated annealing: A near-optimal connection between sampling and counting,'' {\em Journal of the ACM (JACM)}, vol.~56, no.~3, pp.~1--36, 2009.

\bibitem{aaronson2004improved}
S.~Aaronson and D.~Gottesman, ``Improved simulation of stabilizer circuits,'' {\em Physical Review A}, vol.~70, no.~5, p.~052328, 2004.

\bibitem{koenig2014efficiently}
R.~Koenig and J.~A. Smolin, ``How to efficiently select an arbitrary {C}lifford group element,'' {\em Journal of Mathematical Physics}, vol.~55, no.~12, p.~122202, 2014.

\bibitem{bravyi2021hadamard}
S.~Bravyi and D.~Maslov, ``Hadamard-free circuits expose the structure of the {C}lifford group,'' {\em IEEE Transactions on Information Theory}, vol.~67, no.~7, pp.~4546--4563, 2021.

\bibitem{aaronson2020amplitude}
S.~Aaronson and P.~Rall, ``Quantum approximate counting, simplified,'' {\em Symposium on Simplicity in Algorithms}, p.~24–32, Jan 2020.

\bibitem{somma2013spectral}
R.~D. Somma and S.~Boixo, ``Spectral gap amplification,'' {\em SIAM Journal on Computing}, vol.~42, pp.~593--610, 2013.

\bibitem{low2019hamiltonian}
G.~H. Low and I.~L. Chuang, ``Hamiltonian {S}imulation by {Q}ubitization,'' {\em {Quantum}}, vol.~3, p.~163, July 2019.

\end{thebibliography}
\end{document}